\documentclass[sigconf, nonacm]{acmart}

\usepackage{amssymb}
\usepackage{algorithm}
\usepackage[noEnd=true]{algpseudocodex}
\usepackage{amsmath}
\usepackage{booktabs} 
\usepackage{amsthm}
\usepackage{tabularx}
\usepackage{multirow}

\usepackage{listings}
\usepackage{xcolor}
\usepackage{enumitem}

\lstdefinelanguage{cypher}{
  morekeywords={MATCH, RETURN, WITH, WHERE, CREATE, DELETE, DETACH, LIMIT, ORDER, BY, ASC, DESC, UNWIND, AS, DISTINCT, AT, LEAST, MARGINAL, RECALL},
  sensitive=true,
  morecomment=[l]{//},
  morestring=[b]',
  morestring=[b]",
}

\newtheorem{assumption}{Assumption}

\newcounter{SoniaNOC}
\stepcounter{SoniaNOC}

\newcounter{ParisNOC}
\stepcounter{ParisNOC}

\newcounter{XiangyuNOC}
\stepcounter{XiangyuNOC}

\newcounter{FabianNOC}
\stepcounter{FabianNOC}

\newcounter{HenrikNOC}
\stepcounter{HenrikNOC}

\newcounter{VasiaNOC}
\stepcounter{VasiaNOC}

\begin{document}


\title{ConRAD: Conformal Risk-Aware Neural Databases}












\author{Sonia Horchidan}
\affiliation{%
  \institution{KTH, Stockholm}
}
\email{sfhor@kth.se}

\author{Fabian Zeiher}
\affiliation{%
\institution{KTH, Stockholm}
}
\email{zeiher@kth.se}

\author{Xiangyu Shi}
\affiliation{%
  \institution{KTH, Stockholm}
}
\email{xiangyus@kth.se}

\author{Vasiliki Kalavri}
\affiliation{%
  \institution{Boston University}
}
\email{vkalavri@bu.edu}

\author{Henrik Boström}
\affiliation{%
  \institution{KTH, Stockholm}
}
\email{bostromh@kth.se}

\author{Ioannis Kontoyiannis}
\affiliation{%
  \institution{University of Cambridge}
}
\email{ik355@cam.ac.uk}

\author{Paris Carbone}
\affiliation{%
  \institution{KTH, Stockholm}
}
\email{parisc@kth.se}

\begin{abstract}


Querying incomplete knowledge graphs with neural predictors is powerful but dangerous. Errors compound across multi-hop pipelines with no formal bound on the completeness of results. We introduce ConRAD, the first framework to enforce declarative recall guarantees natively within a neural graph database query engine. Given a user-specified risk budget, ConRAD automatically derives per-operator prediction thresholds that satisfy the recall target with finite-sample, distribution-free statistical validity via Conformal Risk Control, while maximizing end-to-end precision. To scale calibration across multi-operator query topologies, we introduce a \emph{quantile-space scalarization} that reduces intractable high-dimensional threshold searches to a single parameter. We further design \emph{the conformal gate}, a novel physical operator that dynamically bypasses neural inference when local graph evidence suffices, eliminating unnecessary model inferences in dense graph regions. Evaluated across three benchmarks and three query topologies, ConRAD strictly satisfies all risk budgets, with empirical recall falling below the target by at most 0.046 across all settings. It reduces neural invocations to zero in near-complete graph regions, and achieves precision that matches or exceeds best-case static baselines that offer no guarantees and require manual threshold search. \looseness=-1

\end{abstract}

\maketitle



\section{Introduction}
\label{sec:intro}

\begin{figure}
    \centering
    \includegraphics[width=0.98\linewidth]{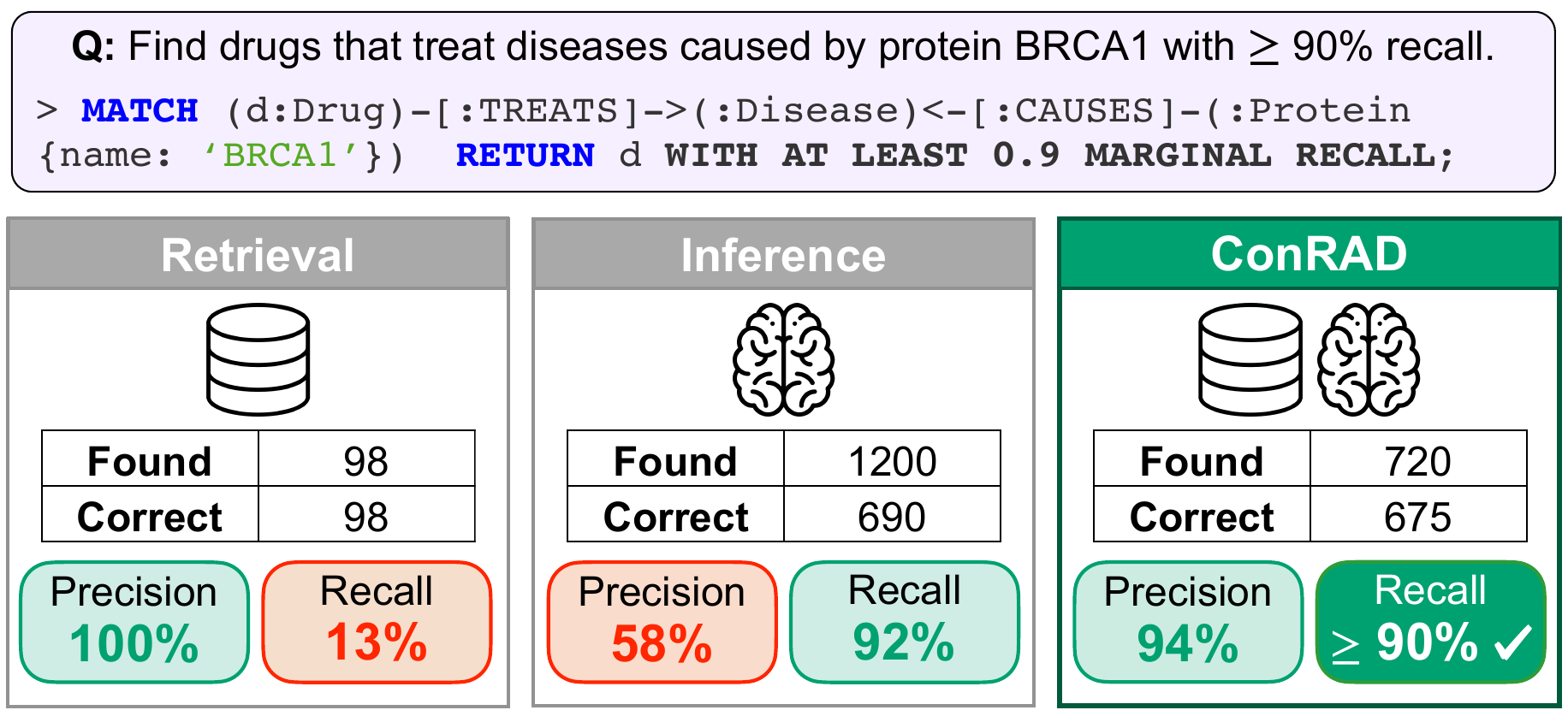}


    \caption{Resolving the precision-recall trade-off on a query 
with 750 ground truth answers. Pure retrieval  fails the recall target (13\% vs.\ $\ge 90\%$) due to graph  incompleteness. Uncalibrated inference meets the recall target (92\%) but the precision collapses to 58\%. ConRAD enforces recall as a hard constraint ($\ge 90\%$) and  treats precision as the optimization objective, achieving 94\% precision while strictly satisfying the risk budget.}
    \label{fig:hook}

\end{figure}

Knowledge graphs (KG) operate under the Open World Assumption (OWA), where missing edges represent unobserved facts rather than true negatives~\cite{hogan2021knowledge, abiteboul1995foundations}. Neural Graph Databases (NGDBs) aim to address this incompleteness by composing exact graph retrievals with probabilistic neural predictors~\cite{DBLP:journals/tmlr/Ren0ZLC24, DBLP:conf/log/BestaISODPCH22, DBLP:conf/vldb/Horchidan23, DBLP:conf/esws/HorchidanC23}. Recent foundation models~\cite{galkin2023towards, DBLP:conf/nips/GalkinZ00Z24} excel at answering complex queries over incomplete KGs. However, these models are trained to optimize pointwise accuracy on individual predictions~\cite{tabacof2020probability}, and while post-hoc calibration techniques can improve score reliability at the single-operator level~\cite{DBLP:conf/kdd/ZhaoK024}, no existing method provides formal correctness guarantees when such predictors are composed into multi-hop query plans. When inference-based operators form multi-hop execution pipelines, statistical errors compound non-linearly. An early false negative irreparably prunes valid subtrees, while a false positive floods downstream operators with cascading noise. In standard machine learning practice, practitioners mitigate uncertainty by manually tuning static prediction thresholds or top-$k$ limits.
While this heuristic approach may suffice for isolated predictions, it breaks down entirely under the compositionality of database queries. A globally tuned threshold cannot dynamically adapt to operator heterogeneity or local graph sparsity. By relying on brittle hyperparameter tuning rather than system-managed bounds, this approach effectively forces the application developer to act as the query optimizer, manually managing risk across complex query topologies without any mathematical assurances.

The consequences are severe in safety-critical deployments. Consider a drug safety monitoring application over a biomedical KG. A query such  as \emph{find all approved drugs that interact with proteins involved in pathways associated with cardiac events} requires multi-hop conjunctive reasoning across drug $\rightarrow$ protein, protein $\rightarrow$ pathway, and pathway $\rightarrow$ adverse effect relations. The underlying graph is inherently incomplete, as thousands of new biological interactions are published in the literature weekly~\cite{chandak2023building}. However, if uncalibrated thresholds cause the system to silently miss valid drug candidates, the entire screening pipeline becomes unreliable. Without formal statistical bounds on expected result completeness, adopting probabilistic neural predictors in such workflows is not viable~\cite{vayena2018machine}.





\begin{table}[t]
\caption{Summary of main notations.}
\label{tab:notations}
\vspace{-1em}
\renewcommand{\arraystretch}{0.9}
\small
\begin{tabularx}{\columnwidth}{l X}
\toprule
\textbf{Symbol} & \textbf{Description} \\ \midrule
$\hat{\mathcal{G}} = (\mathcal{V}, \mathcal{L}, \hat{\mathcal{E}})$ & Observed (incomplete) KG with factual triples $\hat{\mathcal{E}}$. \\
$\mathcal{G} = (\mathcal{V}, \mathcal{L}, \mathcal{E})$ & Complete ground-truth KG under OWA ($\hat{\mathcal{E}} \subseteq \mathcal{E})$. \\
$q = \omega_1 \circ \dots \circ \omega_k$ & Query $q$ as a composition of $k$ operators. \\ 
$\mathcal{T}$ & Query topology (e.g., \texttt{3p}, \texttt{2u}, \texttt{2ip}). \\
$\omega_i$ & The $i$-th operator in the query plan. \\
$f^{(i)}$ & Retrieval over observed triples $\hat{\mathcal{E}}$. \\ 
$\tilde{f}^{(i)}_{\lambda_i}$ & Stochastic inference with threshold $\lambda_i$. \\ 
$\bar{f}^{(i)}_{\lambda_i}$ & Conformal gate with threshold $\lambda_i$. \\
$\hat{Y}^{(i)}$ & Set of output entities produced by operator $\omega_i$. \\ 
$Y^{(i)}$ & Ground truth set (entities reachable in $\mathcal{G}$). \\ 
$\boldsymbol{\lambda} = [\lambda_1, \dots, \lambda_k]$ & Vector of thresholds for all operators in query $q$. \\ 
$\mathcal{R}(\boldsymbol{\lambda})$ & Expected end-to-end risk (e.g., FNR) under $\boldsymbol{\lambda}$. \\
$\mathcal{C}(\boldsymbol{\lambda})$ & Expected execution cost (e.g., cardinality) under $\boldsymbol{\lambda}$. \\
$\boldsymbol{\gamma} \in \Gamma$ & Scalarization strategy. \\
$\alpha$ & User-specified risk budget (e.g., maximum FNR). \\ 
$\eta \in [0,1]$ & Global scalar parameter for pipeline strictness. \\ 
$\delta \in (0,1)$ & Routing threshold in the conformal gate. \\\bottomrule
\end{tabularx}
\vspace{-2em}
\end{table}

We argue that correctness must transition from developer-tuned hyperparameters to a declarative constraint managed natively by the query engine. In a traditional database system, the user specifies a logical query and the optimizer selects the fastest valid physical plan. Similarly, a neural database should allow users to specify a risk budget, delegating the selection of the most precise execution plan to the system. To realize this vision, we introduce ConRAD (\underline{Con}formal \underline{R}isk-\underline{A}ware \underline{D}atabases). To our knowledge, ConRAD is the first framework to provide formal correctness guarantees over composed stochastic query plans in an NGDB. Given a recall target (e.g., risk budget $\alpha=0.1$, requiring 90\% recall), ConRAD automatically derives per-operator prediction thresholds that jointly satisfy the global risk budget while maximizing end-to-end precision. Our approach is grounded in Conformal Risk Control (CRC)~\cite{DBLP:conf/iclr/AngelopoulosBFL24}, which guarantees finite-sample statistical validity for black-box models without relying on unrealistic data distribution assumptions.

To illustrate this tension, consider the precision–recall trade-off depicted in Figure~\ref{fig:hook}. We execute a query over an incomplete KG where the ground truth contains 750 answers. Pure database retrieval yields perfect precision but recovers only 13\% of answers, because graph incompleteness prevents most valid paths from being explored. Replacing deterministic retrieval with uncalibrated neural link prediction recovers more answers but floods the pipeline with false positives. There, while recall incidentally reaches 92\%, the user must accept an unpredictable deterioration in result quality, as precision drops to  58\%. ConRAD bridges this gap by calibrating a hybrid approach that merges retrieval and inference to maximize precision (94\%) while satisfying a user-defined recall target ($\geq 90\%$).

Realizing this vision poses three technical challenges. First, standard CRC calibrates a scalar threshold for a single predictor, whereas NGDB query plans require coordinated calibration across multiple operators. Naively partitioning the global risk budget via the union bound ignores inter-operator correlations and is wasteful, as we demonstrate in Sec.~\ref{sec:method}. Second, the joint threshold space grows exponentially in the number of operators and is non-decomposable: each operator's threshold shifts the input distribution of every downstream operator, so independent tuning cannot guarantee the joint constraint. Per-query guarantees are unattainable with finite calibration data~\cite{foygel2021limits}. Third, neural inference is both computationally expensive and fundamentally unnecessary when the observed graph is locally complete. Invoking a neural model on dense, well-connected regions wastes resources and introduces false positives that degrade downstream precision. These challenges motivate the design of ConRAD. Our contributions are as follows:
\begin{enumerate}
\item We formalize the mapping of user-declared recall targets to per-operator threshold vectors as a constrained optimization problem, replacing ad-hoc threshold tuning with a system-managed calibration process grounded in finite-sample statistical guarantees. The formulation is agnostic to both the neural scoring model and the underlying graph store~(Sec.~\ref{sec:problem}).

\item We extend Conformal Risk Control from scalar to vector calibration over multi-operator query topologies via a \emph{quantile-space scalarization} that reduces the $k$-dimensional threshold space to a single monotonic parameter, while preserving the nested sets required by CRC~(Sec.~\ref{sec:vector-crc}).

\item We introduce \emph{the conformal gate}, a novel physical database operator that dynamically routes between retrieval-based execution and neural inference at each logical operator. By placing exact and neural evidence on a unified calibrated scale, the gate bypasses inference entirely in dense graph regions~(Sec.~\ref{sec:gate}).

\item We implement what is, to our knowledge, the first fully queryable neural graph database with formal correctness guarantees, integrating UltraQuery~\cite{DBLP:conf/nips/GalkinZ00Z24} with Neo4j~\footnote{https://neo4j.com/}. We evaluate on three benchmarks, three query topologies, and data incompleteness levels ranging from 5\% to 40\%. ConRAD satisfies recall targets across all settings (max. downward deviation: 0.046) and achieves precision matching or exceeding best-case static baselines that offer no guarantees. Furthermore, ConRAD reduces neural invocations by up to 100\% in well-connected regions under generous risk budgets~(Sec.~\ref{sec:experiments}).
\end{enumerate}



Although our evaluation centers on KGs, ConRAD's formal guarantees extend to any predictive pipeline that composes probabilistic operators into a directed acyclic graph, including retrieval-augmented generation (RAG)~\cite{lewis2020retrieval}, multi-stage retrieval~\cite{khattab2023dspy, patel2024semantic}, and learned database components~\cite{DBLP:conf/sigmod/KraskaBCDP18, DBLP:conf/cidr/KaranasosIPSPPX20, DBLP:conf/sigmod/ParkSBSIK22}. By transitioning correctness from a manually-tuned parameter to a declarative system constraint, ConRAD provides a principled blueprint for integrating machine learning inference into the broader data processing stack.

\section{Setting and Foundations}
\label{sec:overview}

ConRAD targets the execution of complex logical queries over incomplete KGs under OWA. Given a user query and a declarative recall target, ConRAD automatically derives calibrated per-operator thresholds that satisfy the guarantee while maximizing end-to-end precision. Below, we describe ConRAD's data model, supported queries, execution model, and the statistical foundations it relies on. Table~\ref{tab:notations} summarizes the notation used throughout the paper.

\subsection{Data Model and Supported Queries}
\label{sec:data-model}

We represent an observed KG as a directed, labeled graph $\hat{\mathcal{G}} = (\mathcal{V}, \mathcal{L}, \hat{\mathcal{E}})$, where $\mathcal{V}$ is the set of entities, $\mathcal{L}$ the set of relation types (edge labels), and $\hat{\mathcal{E}} \subseteq \mathcal{V} \times \mathcal{L} \times \mathcal{V}$ the set of observed triples. We adopt the standard KG triple model used by existing neural query answering frameworks~\cite{ren2020query2box, ren2020beta}. Node and edge properties, as found in property graph models, can be represented by encoding property values as entities connected via dedicated relation types. Under OWA, we assume a complete ground-truth graph $\mathcal{G} = (\mathcal{V}, \mathcal{L}, \mathcal{E})$ with $\hat{\mathcal{E}} \subseteq \mathcal{E}$, with $\mathcal{E} \setminus \hat{\mathcal{E}}$ representing true but unobserved facts.

ConRAD supports the Existential Positive First-Order (EPFO) fragment of first-order logic, supporting existential quantification ($\exists$), conjunction ($\land$), and disjunction ($\lor$)~\cite{ren2020query2box}. Queries are DAGs of operators where leaf nodes are anchor entities and the sink produces the answer set. Three operators map directly to EPFO primitives: \textit{projection} navigates a relation $r \in \mathcal{L}$ from a source set $S \subseteq \mathcal{V}$
to retrieve reachable entities; \textit{intersection} and \textit{union} perform set conjunction and disjunction over intermediate results.

We evaluate over three standard query topologies~\cite{ren2020query2box, ren2020beta, DBLP:conf/nips/GalkinZ00Z24} that isolate the primary composition primitives: (1) \texttt{3p} (three-hop projection), a chain that tests cascading error propagation (analogous to a join operation), defined as $h \xrightarrow{r_1} t_1 \xrightarrow{r_2} t_2 \xrightarrow{r_3}\;?$; (2) \texttt{2u} (two-hop union), a parallel topology that maximizes recall by design, defined as $(h_1 \xrightarrow{r_1}\;?) \cup (h_2 \xrightarrow{r_2}\;?)$; and (3) \texttt{2ip} (intersect-project), defined as $(h_1 \xrightarrow{r_1}\;?) \cap (h_2 \xrightarrow{r_2}\;?) \xrightarrow{r_3}\;?$, which merges independent evidence streams.

\noindent
\textbf{What ConRAD does not support.} ConRAD excludes negation and iterative (recursive) queries. Negation is semantically ill-defined under OWA~\cite{abiteboul1995foundations} and breaks the monotonicity of operator composition (Theorem~\ref{thm:monotonicity}), invalidating the nestedness property on which our recall guarantees depend (Sec.~\ref{sec:vector-crc}). Iterative queries are excluded because the threshold vector $\boldsymbol{\lambda}$ we propose in Sec.~\ref{sec:vector-crc} has fixed dimensionality determined by the query topology. Unbounded recursion would, therefore, require a variable-length threshold vector, which is incompatible with our offline calibration procedure.

\subsection{Execution Model}
\label{sec:execution-model}

A key design principle of ConRAD is that each logical operator is backed by two execution primitives: a deterministic \emph{retrieval} operator that traverses observed edges in $\hat{\mathcal{G}}$, and a stochastic \emph{inference} operator that scores candidate triples using a neural model. The conformal gate (Sec.~\ref{sec:gate}) dynamically composes these primitives.

\smallskip\noindent
\textit{Retrieval.} The retrieval operator $f^{(i)}$ executes an exact traversal over $\hat{\mathcal{E}}$. Given an input set $\hat{Y}^{(i-1)} \subseteq \mathcal{V}$ and a relation $r \in \mathcal{L}$:
\begin{equation}\label{eq:retrieval-op}   
f^{(i)}(\hat{Y}^{(i-1)}, r)   = \{t \in \mathcal{V} \mid \exists\, h \in \hat{Y}^{(i-1)},\;       (h, r, t) \in \hat{\mathcal{E}}\} 
\end{equation} 
By construction, every returned triple exists in $\hat{\mathcal{G}}$, guaranteeing zero false positives. However, recall is bounded by graph completeness.

\smallskip\noindent
\textit{Inference.} The inference operator $\tilde{f}^{(i)}_{\lambda_i}$ replaces exact traversal with a learned scoring function $\phi : \mathcal{V} \times \mathcal{L} \times \mathcal{V} \to [0,1]$ provided by a neural model. Given an input set $\hat{Y}^{(i-1)}$ and a relation $r$, the operator admits all candidate entities whose score exceeds a threshold $\lambda_i$:
\begin{equation}\label{eq:inference-op}
  \tilde{f}^{(i)}_{\lambda_i}(\hat{Y}^{(i-1)}, r)
  = \{t \in \mathcal{V} \mid \exists\, h \in \hat{Y}^{(i-1)},\;
      \phi(h, r, t) \geq \lambda_i\}
\end{equation}
The threshold $\lambda_i$ governs the precision–recall trade-off at each operator: lowering $\lambda_i$ recovers missing answers but potentially floods downstream operators with false positives.

\noindent
\textbf{The threshold selection problem.} In practice, per-operator thresholds are typically selected manually, following the standard workflow of deploying ML models with score-based filtering. This heuristic approach offers no runtime guarantees and is inadequate for composed query plans. Neural scores are often uncalibrated (i.e., a softmax score of 0.9 does not imply 90\% correctness), their distributions shift across operators and graph neighborhoods, and errors compound non-linearly. The resulting optimization problem is inherently coupled: each operator's threshold affects the input distribution of every downstream operator, making independent per-operator tuning insufficient. ConRAD eliminates this burden. Given a query $q = \omega_1 \circ \cdots \circ \omega_k$ represented as a DAG of $k$ operators, ConRAD automatically derives a threshold vector $\boldsymbol{\lambda} = [\lambda_1, \dots, \lambda_k]$ that satisfies a user-declared recall target while maximizing precision. \looseness=-1


\subsection{Statistical Foundations}
\label{sec:stat-foundations}

ConRAD's calibration procedure builds on Conformal prediction methods, which we briefly review in this section.


\subsubsection{Standard Split Conformal Prediction.} Intuitively, the Conformal Prediction (CP) framework calibrates a model by measuring how unusual new data points are compared to a reference set~\cite{vovk2005algorithmic, DBLP:journals/jmlr/ShaferV08, DBLP:journals/ftml/AngelopoulosB23}. In its split variant, CP avoids model retraining by partitioning the data into a held-out calibration set.


Consider a calibration dataset of $n$ exchangeable tuples, $\mathcal{D}_{\text{cal}} = \{(x_i, y_i)\}_{i=1}^n$. We first define a non-conformity score $s(x, y) \in \mathbb{R}$, which quantifies the disagreement between a tuple's features $x$ and a proposed label $y$ from a finite label set $\mathcal{Y}$. For a learned filtering predicate, this could simply be, for instance, one minus the neural confidence, $1 - \phi(x, y)$. Standard Conformal Prediction computes the empirical quantile $\hat{q}$ of these non-conformity scores on $\mathcal{D}_{\text{cal}}$ at the level $\lceil (n+1)(1-\alpha) \rceil / n$, where $\alpha$ is the user-defined risk budget. For any new, unseen test tuple $(x_{n+1}, y_{n+1})$ assumed to be exchangeable with the calibration tuples, where $x_{n+1}$ is observed and $y_{n+1}$ is the unknown ground truth, we construct a prediction set $C(x_{n+1})$ containing all possible labels whose non-conformity scores fall below this threshold:
\begin{equation}
C(x_{n+1}) = \{ y \in \mathcal{Y} \mid s(x_{n+1}, y) \leq \hat{q} \}
\end{equation}

This procedure provides a \textit{marginal coverage guarantee}: the true outcome is guaranteed to be in the prediction set with high probability, formally $\mathbb{P}(y_{n+1} \in C(x_{n+1})) \geq 1 - \alpha$. However, standard CP controls only miscoverage (0-1 loss). Database workloads require guarantees over aggregate, set-valued metrics (e.g., recall, False Negative Rate).


\subsubsection{Conformal Risk Control.} To accommodate these system-level metrics, Conformal Risk Control (CRC)~\cite{DBLP:conf/iclr/AngelopoulosBFL24} generalizes the standard CP framework to bound arbitrary loss functions $L(C_\lambda(x), y) \in [0, B]$. CRC introduces a tunable parameter $\lambda$ that dictates the "strictness" or size of the prediction set $C_\lambda(x)$. Crucially, the chosen loss function must be non-increasing with respect to $\lambda$. For example, as a selection operator becomes more permissive (larger $\lambda$), the False Negative Rate stays the same or decreases. The goal of CRC is to identify a threshold $\hat{\lambda}$ that guarantees the expected risk (e.g., FNR) remains strictly below a user-defined risk budget $\alpha$. CRC achieves this by finding the tightest parameter that satisfies a corrected empirical risk bound on the calibration set:
\begin{equation} \label{eq:crc}
\hat{\lambda} = \inf \{ \lambda : \frac{n}{n+1} \widehat{R}(\lambda) + \frac{B}{n+1} \leq \alpha \}
\end{equation}
where $\widehat{R}(\lambda) = \frac{1}{n} \sum_{i=1}^n L(C_\lambda(x_i), y_i)$ is the empirical risk over $\mathcal{D}_{\text{cal}}$ and $B$ is an upper bound of the loss. The correction terms account for the finite size of the calibration data, guaranteeing that the expected risk on unseen queries remains bounded: 
\begin{equation}
\mathbb{E}[L(C_{\hat{\lambda}}(x_{n+1}), y_{n+1})] \leq \alpha
\end{equation}

CRC calibrates a single operator. Extending it to composed query plans without massive precision degradation is the core challenge we address in Sec.~\ref{sec:vector-crc}.

\section{Problem Formulation}\label{sec:problem}

\begin{figure}[t]
    \centering
    \includegraphics[width=1\linewidth]{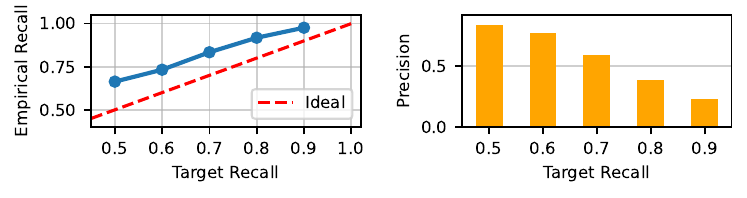}
    \caption{The cost of independent calibration on a 3-hop query (FB15k-237). (Left) The naive union bound approach consistently overshoots target recall, underutilizing the global risk budget. (Right) Consequently, excessively permissive thresholds admit cascading false positives, plunging end-to-end precision below 40\% for all target recalls.}
    \label{fig:ub}
    \vspace{-1.5em}
\end{figure}

We formalize the execution of EPFO queries over a Neural Graph Database as a constrained optimization problem, replacing ad-hoc threshold tuning with a system-managed calibration process.
Let $q = \omega_1 \circ \cdots \circ \omega_k$ be a query represented as a DAG of $k$ operators. The execution of the query is governed by a threshold vector $\boldsymbol{\lambda} = [\lambda_1, \dots, \lambda_k] \in [0, 1]^k$, where each $\lambda_i$ is assigned to a distinct inference operator $\tilde{f}^{(i)}$ in the query DAG (under any fixed logical ordering). To ensure a uniform threshold interface, we assume each inference operator $\tilde{f}^{(i)}$ produces a normalized score in the range $[0, 1]$ (e.g., via a softmax or sigmoid activation). We assume queries are drawn from a distribution $\mathcal{D}$. We define the query-level ground truth $Y^{(k)} \subseteq \mathcal{V}$ as the set of entities reachable over the complete graph $\mathcal{G}$, and denote by $\hat{Y}^{(k)}_{\boldsymbol{\lambda}}$ the set of entities retrieved by the final operator in the pipeline under threshold vector $\boldsymbol{\lambda}$.

\noindent \textbf{Risk Constraint.}  We define the risk function $\mathcal{R}(\boldsymbol{\lambda})$ as the expected end-to-end False Negative Rate (FNR) over $\mathcal{D}$. A threshold vector $\boldsymbol{\lambda}$ is valid if and only if the expected FNR remains below a user-specified risk budget $\alpha \in (0, 1)$:
%
\begin{equation}
\label{eq:risk}
\mathcal{R}(\boldsymbol{\lambda}) = \mathbb{E}\left[1 - \frac{|\hat{Y}^{(k)}_{\boldsymbol{\lambda}} \cap Y^{(k)}|}{|Y^{(k)}|}\right] \leq \alpha
\end{equation}


In a query topology of monotonic operations (selections, projections, joins), errors are asymmetric. A false positive adds downstream overhead but can be filtered by subsequent operators. A false negative is irrecoverable: once a valid tuple is pruned, no downstream join can restore it. By treating FNR as a hard constraint, ConRAD preserves the database expectation of completeness while delegating precision maximization to the cost objective. Equivalently, bounding $\text{FNR} \leq \alpha$ guarantees recall
$\geq 1 - \alpha$. We use both formulations interchangeably, referring to $\alpha$ as the \emph{risk budget} and to $1 - \alpha$ as the \emph{recall target}.


\noindent \textbf{Optimization Objective.} Among all valid threshold vectors $\boldsymbol{\lambda}$, the optimizer seeks the plan that maximizes result quality. We minimize expected cardinality of the final answer set:
%
\begin{equation}
  \mathcal{C}(\boldsymbol{\lambda}) = \mathbb{E}\left[|\hat{Y}^{(k)}_{\boldsymbol{\lambda}}|\right]  
\end{equation}

Since the risk constraint already guarantees sufficient true positive coverage, minimizing $|\hat{Y}^{(k)}_{\boldsymbol{\lambda}}|$ directly minimizes false positives, maximizing end-to-end precision.



\noindent \textbf{Risk-Constrained Optimization.} The objective of our optimization framework is to identify the threshold vector $\boldsymbol{\lambda}$ that minimizes $\mathcal{C}$ without violating the risk constraint. Formally:
\begin{equation}
\label{eq:objective}
\begin{aligned}
\min_{\boldsymbol{\lambda} \in [0,1]^k} \quad & \mathcal{C}(\boldsymbol{\lambda}) \\
\text{subject to} \quad & \mathcal{R}(\boldsymbol{\lambda}) \leq \alpha, \quad \alpha \in (0, 1)
\end{aligned}
\end{equation}

This formulation differs from classical cost-based query optimization, where all physical plans produce correct and complete results and the optimizer aims to minimize the execution time. In our setting, the answer set itself is a function of $\boldsymbol{\lambda}$: the optimizer must navigate the continuous trade-off between selectivity and coverage, finding
the tightest threshold vector that still satisfies the recall target. We treat the risk constraint as hard and non-negotiable, backed by finite-sample statistical validity; the cost objective is best-effort.


\section{The ConRAD Framework}
\label{sec:method}


We first show that independent per-operator calibration is provably wasteful, then introduce a scalarization strategy that reduces the $k$-dimensional threshold search to a single parameter, and finally design the conformal gate operator that dynamically routes between retrieval and inference.

\noindent \textit{The Pessimism of Independent Calibration.}
A naive baseline for calibrating a $k$-operator EPFO query allocates the global risk budget $\alpha$ equally, assigning $\alpha/k$ to each operator and bounding the total risk via the union bound ($\mathcal{R}(\lambda) \le \sum_{i=1}^{k} \alpha/k$). While statistically valid, the union bound assumes worst-case dependence between operators, ignoring the correlations between operators. The effect is systematic over-allocation, as each operator receives a more generous threshold than necessary, and the surplus false positives compound across hops. Figure~\ref{fig:ub} illustrates this for a \texttt{3p} query topology on FB15k-237~\cite{DBLP:conf/nips/BordesUGWY13}. At a 50\% recall target, independent calibration overshoots to 0.66 empirical recall, demonstrating an inability to constrain cardinality growth. At a 90\% target, the thresholds provide no meaningful pruning: empirical recall reaches 0.976 while precision collapses to 0.23. This motivates the joint, pipeline-aware calibration strategy we propose next.

\begin{figure*}[t]
    \centering
    \begin{minipage}[t]{0.66\textwidth}
        \centering
        \includegraphics[width=\textwidth]{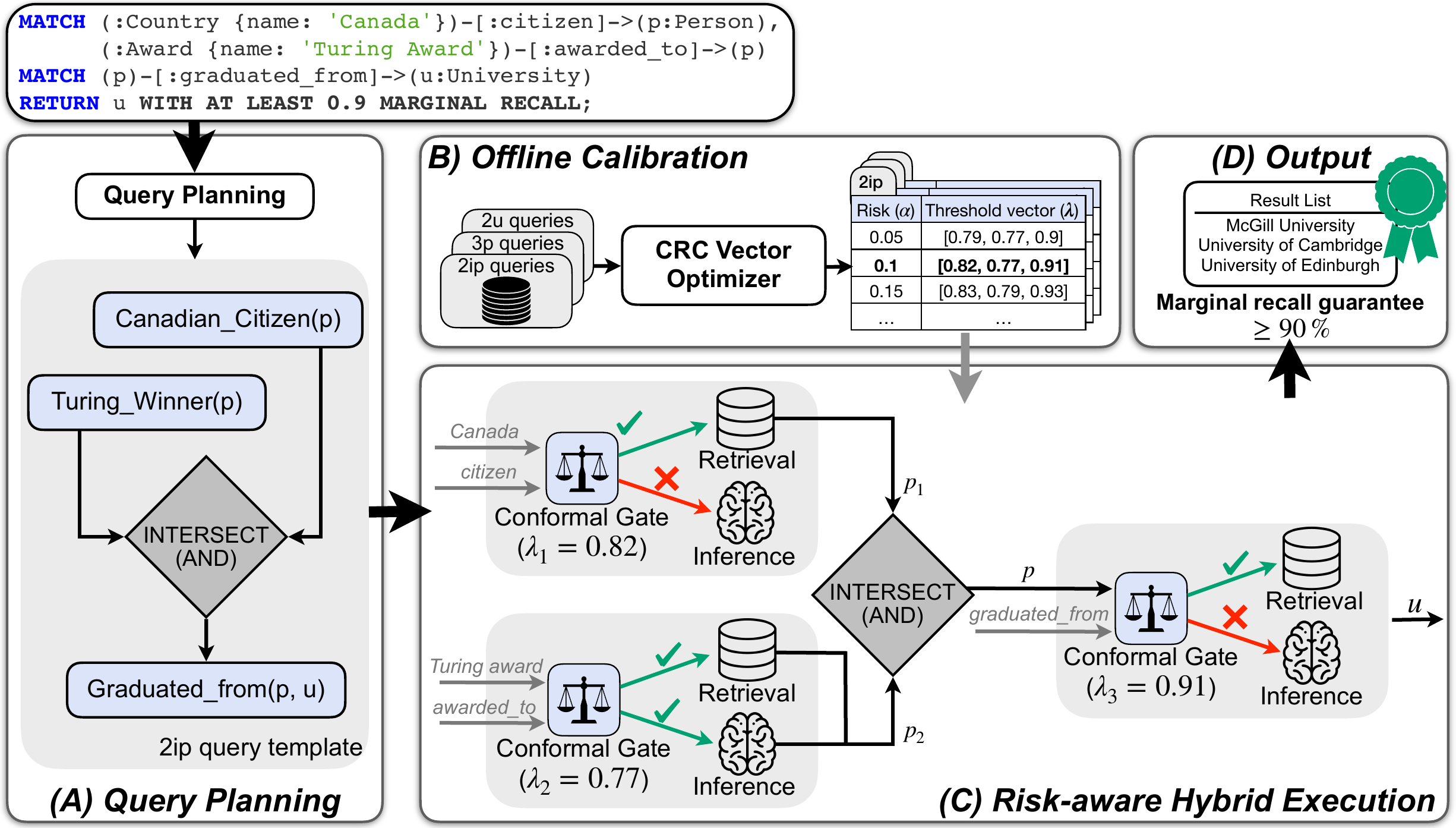}
        \caption{ConRAD overview.} 
        \label{fig:overview}
    \end{minipage}%
    \hfill
    \begin{minipage}[t]{0.33\textwidth}
        \centering
        \includegraphics[width=\textwidth]{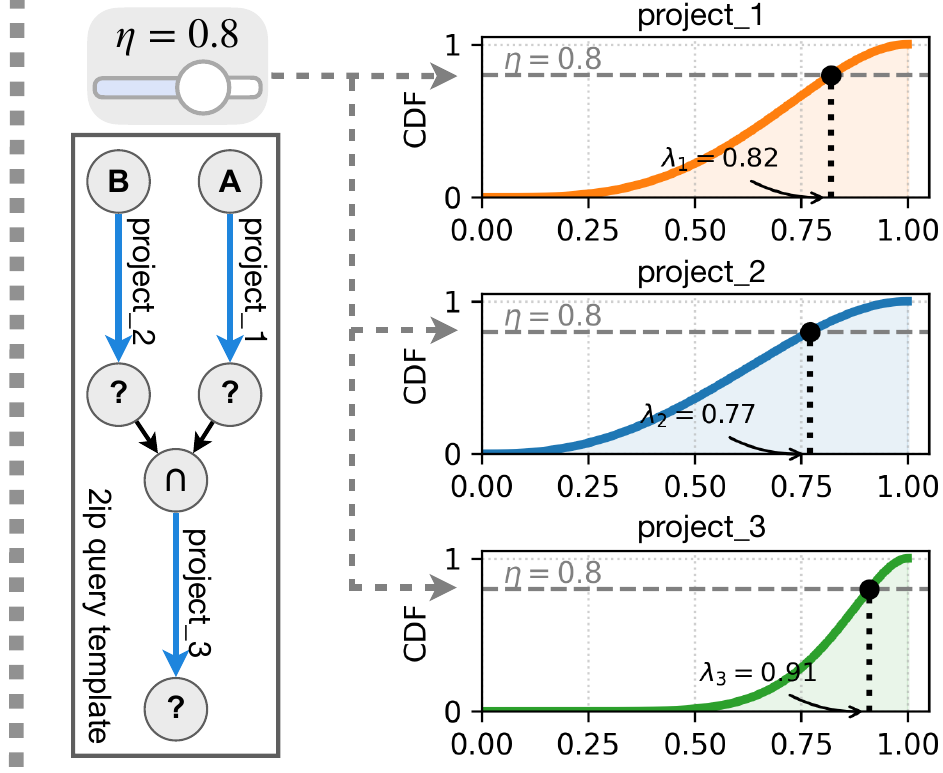}
        \caption{\texttt{2ip} scalarization example.}
        \label{fig:optimizer}

    \end{minipage}
    \vspace{-1em}
\end{figure*}

\subsection{Conformal Calibration for Composed Plans}
\label{sec:vector-crc}

Calibrating $k$ thresholds jointly is both expensive and theoretically incompatible with CP. A grid search over $[0,1]^k$ is intractable for any practical optimizer. Furthermore, CRC requires that lowering a threshold can only grow the result set (nestedness), but moving in multiple dimensions at once can shrink one operator's output while expanding another's, breaking this property. Our solution is to collapse the search to one dimension: we derive the full threshold vector $\boldsymbol{\lambda}$ from a single scalar parameter $\eta \in [0,1]$.


\subsubsection{Scalarization Strategies}
\label{sec:scalarization} We define a scalarization as a mapping $\mathcal{M}: [0,1] \to [0,1]^k$ that derives the threshold vector $\boldsymbol{\lambda}$ from a single global parameter $\eta$. Any component-wise monotonic mapping preserves the nestedness required by CRC, making it statistically valid. A key challenge is that different operators produce confidence scores on incomparable scales. For instance, a score of 0.7 may represent high confidence for one predictor and near-random guessing at another. To resolve this, our scalarization operates in quantile space. Figure~\ref{fig:optimizer} illustrates this on a \texttt{2ip} topology. A single global parameter $\eta = 0.8$ is routed through each operator's empirical quantile function $\hat{Q}_j$, derived from the calibration set, producing per-operator thresholds ($\lambda_1 = 0.82$, $\lambda_2 = 0.77$, $\lambda_3 = 0.91$) that normalize across incomparable score distributions. Each threshold can further be scaled by a per-operator exponent $\gamma_j > 0$:
\begin{equation} \label{eq:scalarization}
   \lambda_j(\eta) = \hat{Q}_j(\eta^{\gamma_j}), \quad j = 1, \ldots, k 
\end{equation}

The quantile function $\hat{Q}_j$ normalizes across operators with different score distributions, ensuring consistent selectivity regardless of raw score range. The exponent $\gamma_j$ controls how aggressively the $j$-th operator tightens relative to others: a \emph{loose-early, tight-late} $\boldsymbol{\gamma}$ can preserve recall at initial hops to maximize final precision, while a \emph{tight-early, loose-late} $\boldsymbol{\gamma}$ can aggressively prune intermediate cardinalities, analogous to predicate pushdown. Since any strictly positive $\boldsymbol{\gamma}$ yields a component-wise monotonic mapping, all such choices preserve nestedness and are valid under the conformal prediction framework, as we formally prove below.


\subsubsection{Theoretical Validity} 
\label{sec:validity}

We prove that the nested prediction sets required by CRC
hold globally across composed query plans for any monotonic
scalarization.
\begin{assumption}[Pointwise Scoring]
\label{asm:pointwise}
The inference operator $\tilde{f}^{(i)}_{\lambda_i}$ evaluates each candidate triple $(h, r, t)$ independently. The inclusion of entity $t$ depends solely on whether $\phi(h, r, t) \geq \lambda_i$.
\end{assumption}

This assumption holds for most standard KG scoring functions~\cite{bordes2013translating, trouillon2016complex, sun2019rotate, DBLP:conf/nips/GalkinZ00Z24}, which compute $\phi(h, r, t)$ as a function of the triple alone. Each inference operator thus evaluates a per-tuple selection predicate. A candidate is admitted if and only if its score exceeds $\lambda_i$, independently of the scores assigned to other candidates. This holds by construction for all threshold-based operators but not for rank-based selections (e.g., top-$k$).


\begin{theorem}[Global Monotonicity] \label{thm:monotonicity}
Let $q$ be an EPFO query satisfying Assumption~\ref{asm:pointwise}. If the configuration $\boldsymbol{\lambda}(\eta)$ is determined by any component-wise non-decreasing scalarization $\mathcal{M}$, then for any
$\eta_a < \eta_b$:
$$\hat{Y}^{(k)}_{\boldsymbol{\lambda}(\eta_b)} \subseteq
  \hat{Y}^{(k)}_{\boldsymbol{\lambda}(\eta_a)}$$
  
Consequently, $\mathcal{R}(\boldsymbol{\lambda})$ is
non-decreasing in $\eta$.
\end{theorem}
\begin{proof}
By induction over the DAG in topological order.

\noindent \emph{Base case:} For leaf operators, the input is a fixed anchor set independent of $\eta$. Since $\mathcal{M}$ is component-wise non-decreasing, $\eta_a < \eta_b$ implies $\lambda_j(\eta_a) \leq \lambda_j(\eta_b)$ for all $j$. By Assumption 1, raising a threshold can only shrink a neural operator's output, so $\tilde{f}^{(j)}_{\lambda_j(\eta_b)}(S, r) \subseteq \tilde{f}^{(j)}_{\lambda_j(\eta_a)}(S, r)$ for any fixed input set $S$.

\noindent \emph{Inductive step:}  Assume the subset property holds for all predecessors of $\omega_i$ in topological order. Every EPFO operator is input-monotone: if $A \subseteq B$ then $\omega_i(A) \subseteq \omega_i(B)$, which holds for projection (fewer sources yield fewer targets), intersection ($A_j \subseteq B_j$ implies $\bigcap A_j \subseteq \bigcap B_j$), and union ($A_j \subseteq B_j$ implies $\bigcup A_j \subseteq \bigcup B_j$). For projection operators, the stricter threshold further shrinks the output. Combined, $\hat{Y}^{(i)}_{\boldsymbol{\lambda}(\eta_b)} \subseteq \hat{Y}^{(i)}_{\boldsymbol{\lambda}(\eta_a)}$ holds at stage $i$, and by induction to the final output $\hat{Y}^{(k)}$.

\noindent Risk monotonicity follows: since the ground truth $Y^{(k)}$ is fixed, a smaller prediction set can only capture fewer true positives, so $\mathcal{R}(\boldsymbol{\lambda})$ is non-decreasing in $\eta$.
\end{proof}

This nestedness enables tractable calibration. CRC performs a monotone search over $\eta$ to find the tightest threshold vector satisfying $\mathcal{R}(\boldsymbol{\lambda}) \leq \alpha$, with the finite-sample guarantee of Eq.~\ref{eq:crc} applying directly. All scalarizations in the family of Sec.~\ref{sec:scalarization} are valid by construction. Importantly, the search finds the optimum along a given trajectory but not necessarily the global minimum of $\mathcal{C}(\boldsymbol{\lambda})$. This relaxation is a necessary trade-off for tractability, analogous to heuristic plan-space pruning in classical query optimizers.



\begin{figure}[t]
    \centering
    \includegraphics[width=\linewidth]{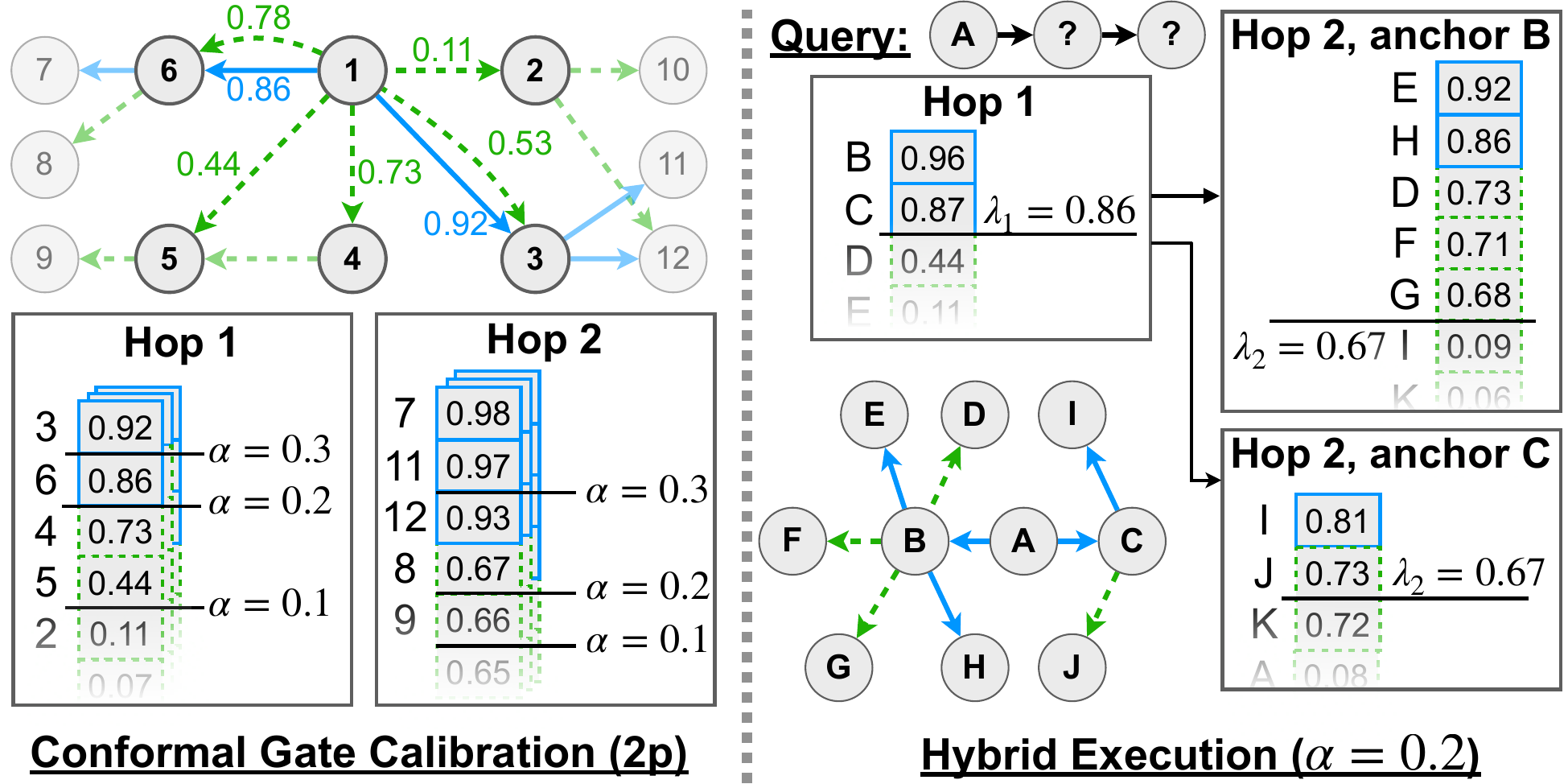}
    \caption{The conformal gate on a two-hop query topology. \textbf{Left:} Calibration maps risk budgets to per-operator thresholds; tighter budgets force inclusion of uncertain neural predictions (green). \textbf{Right:} Online execution at $\alpha = 0.2$ demonstrates adaptive routing. Node A's neighborhood is dense enough for retrieval-only execution (blue), while nodes B and C require hybrid mode to meet the recall target.}
    \label{fig:operator}
\end{figure}

\subsection{The Conformal Gate}
\label{sec:gate}
Invoking a neural model at every hop is expensive and injects unnecessary false positives. We therefore map each logical operator $\omega_i$ to a \emph{conformal gate}, a physical operator that dynamically routes between deterministic retrieval $f^{(i)}$ and neural inference $\tilde{f}^{(i)}_{\lambda_i}$ based on the calibrated threshold $\lambda_i$.



For the calibrated threshold $\lambda_i$ to act as a universal gate, retrieval results (correct by construction) and neural predictions (uncertain) must be scored on a single comparable scale. A naive approach assigning confidence $1.0$ to all retrieval facts creates a mass of tied scores, preventing the calibrator from selecting fine-grained thresholds, as it must either accept or reject all retrieval facts at once. We resolve this by partitioning the score space $[0, 1]$ at a
routing threshold $\delta \in (0, 1)$:
\begin{equation}
\label{eq:unified-score}
S(\tau) =
\begin{cases}
\delta + (1 - \delta) \cdot \mathcal{H}(\tau) & \text{if } \tau \in f^{(i)}(\hat{\mathcal{G}}) \\
\phi(\tau) \cdot (\delta - \epsilon) & \text{otherwise} 
\end{cases}
\end{equation}
where $\tau = (h, r, t)$ is a candidate triple, $\epsilon > 0$ is a negligible numerical margin, and $\mathcal{H}$ is a consistent hash function mapping triples to random scalar values in $[0,1]$. Retrieval facts receive scores in $[\delta, 1]$, while neural predictions are strictly bounded to $[0, \delta - \epsilon)$. This guarantees that the calibrator prioritizes explicit facts over uncertain predictions. $\mathcal{H}$ spreads retrieval scores uniformly across the upper interval, breaking ties and allowing the calibrator to select thresholds at arbitrary granularity, similar to tie-breaking mechanisms used in standard CP methods~\cite{vovk2005algorithmic, DBLP:journals/ftml/AngelopoulosB23}. The tie-breaking enables subsampling within the retrieval-only regime to prevent downstream cardinality blowup. Because CRC calibrates over the joint scores, the choice of $\delta$ does not affect accept/reject behavior.

At runtime, the relationship between the calibrated threshold $\lambda_i$ and the routing threshold $\delta$ determines the gate's execution mode:

\noindent \underline{\textit{Retrieval-only ($\lambda_i \geq \delta$):}} The recall target is satisfied by $\hat{\mathcal{G}}$. The gate executes only $f^{(i)}$ and subsamples retrieval facts where $S(\tau) \geq \lambda_i$, pruning intermediate cardinalities with zero inference cost.

\noindent \underline{\textit{Hybrid ($\lambda_i < \delta$):}} Local graph density is insufficient. The gate returns all retrieval facts and invokes $\tilde{f}^{(i)}_{\lambda_i}$ to recover missing answers.

Figure~\ref{fig:operator} illustrates both modes on a two-hop query topology. At $\alpha = 0.2$, hop 1 calibrates $\lambda_1 = 0.86$: nodes B and C are admitted via retrieval alone. At hop 2, the gate routes adaptively per anchor. Node A's neighborhood is dense enough for retrieval-only execution, while nodes B and C require hybrid mode to meet the recall target.


\section{Implementation}
\label{sec:framework}

We now describe the end-to-end implementation of ConRAD, bridging the theoretical guarantees of Sec.~\ref{sec:method} with a practical query execution engine. ConRAD is implemented as a middleware layer coordinating a deterministic graph store (exposing a standard traversal API) with an external neural scoring service ($\phi(h, r, t) \to [0, 1]$); either component can be swapped independently. Figure~\ref{fig:overview} illustrates the workflow. A user submits a query alongside a declarative recall target. The system (A) parses the query into a logical DAG and maps it to a known query topology, (B) retrieves the pre-computed calibration parameters for that topology and risk budget $\alpha$, (C) compiles the logical plan into a physical DAG of conformal gates, and (D) returns the result set with a marginal recall guarantee. We detail each stage below, beginning with the offline calibration process that provides the statistical foundation for runtime execution.

\begin{algorithm}[t]
\caption{Calibration Phase (Offline)}
\label{alg:calibration}
\begin{algorithmic}[1]
\Require Risk budget $\alpha$, calibration queries $\mathcal{D}_\text{cal}$, query topology $\mathcal{T}$ with $k$ inference operators, routing threshold $\delta$, candidate strategies $\Gamma$, cardinality objective $\mathcal{C}$
\Ensure Calibrated threshold vector $\hat{\boldsymbol{\lambda}}$
\State Split $\mathcal{D}_\text{cal}$ into disjoint subsets $\mathcal{D}_\text{opt}$ and $\mathcal{D}_\text{valid}$ 
\State $n \gets |\mathcal{D}_\text{opt}|$
\For{each operator $j \gets 1$ to $k$}
    \State Compute scores on $\mathcal{D}_\text{opt}$ \Comment{Eq.~\ref{eq:unified-score}}
    \State Compute empirical quantile function $\hat{Q}_j$
\EndFor
\State $\text{best\_c} \gets \infty$
\For{each strategy $\boldsymbol{\gamma} \in \Gamma$} \Comment{Evaluate scalarization strategies}
    \For{each $\eta$ in a discretized grid over $[0, 1]$}
        \State $\boldsymbol{\lambda} \gets [\hat{Q}_j(\eta^{\gamma_j})]_{j=1}^k$ \Comment{Scalarization (Eq.~\ref{eq:scalarization})}
        \For{each $(q_i, Y_i) \in \mathcal{D}_\text{opt}$}
            \State Execute $q_i$ under $\boldsymbol{\lambda}$ to obtain $\hat{Y}^{(k)}_{i,\boldsymbol{\lambda}}$
            \State $L_i \gets 1 - |\hat{Y}^{(k)}_{i,\boldsymbol{\lambda}} \cap Y_i| \;/\; |Y_i|$ \Comment{Per-query FNR}
        \EndFor
        \State $\hat{\mathcal{R}}_\eta \gets \frac{1}{n} \sum L_i$
    \EndFor
\State $\hat{\eta} \gets \sup\left\{\eta : \frac{n}{n+1}\hat{\mathcal{R}}_\eta + \frac{B}{n+1} \leq \alpha\right\}$ \Comment{CRC on $\mathcal{D}_\text{opt}$}

\State $\boldsymbol{\lambda}^* \gets [\hat{Q}_j(\hat{\eta}^{\gamma_j})]_{j=1}^k$
    
    \State $c \gets \text{Evaluate } \mathcal{C}(\boldsymbol{\lambda}^*) \text{ over } \mathcal{D}_\text{valid}$

    \If{$c < \text{best\_c}$}
        \State $\hat{\boldsymbol{\lambda}} \gets \boldsymbol{\lambda}^*$; \quad $\text{best\_c} \gets c$
    \EndIf
\EndFor
\State \Return $\hat{\boldsymbol{\lambda}}$
\end{algorithmic}
\end{algorithm}

\subsection{Offline Calibration and Runtime Execution} \label{sec:runtime}
For each query topology $\mathcal{T}$ and risk budget $\alpha$, the system determines the threshold vector $\hat{\boldsymbol{\lambda}}$ that satisfies the recall target while minimizing result-set cardinality $\mathcal{C}$. This phase runs entirely offline, decoupled from the query critical path. Algorithm~\ref{alg:calibration} details the procedure, which proceeds in three stages over a calibration dataset $\mathcal{D}_\text{cal}$ partitioned into disjoint optimization ($\mathcal{D}_\text{opt}$) and evaluation ($\mathcal{D}_\text{valid}$) sets. First, for each inference operator $j$ in the topology, we execute the queries in $\mathcal{D}_\text{opt}$ over the incomplete graph $\hat{\mathcal{G}}$ and collect both retrieval and inference scores, unified via $S(\tau)$ (Eq.~\ref{eq:unified-score}). The routing threshold $\delta$ is fixed across calibration and online execution, ensuring that $\boldsymbol{\lambda}$ induces identical gate behavior in both phases. From these scores, we compute the empirical quantile function $\hat{Q}_j$ for each operator. Second, we evaluate each candidate scalarization strategy $\boldsymbol{\gamma} \in \Gamma$ (Sec.~\ref{sec:scalarization}), a small predefined set of per-operator exponent vectors, over a discretized grid of 100 values $\eta \in [0, 1]$. At each grid point, the scalarization (Eq.~\ref{eq:scalarization}) maps $\eta$ to a threshold vector $\boldsymbol{\lambda}$, and the full pipeline is executed over $\mathcal{D}_\text{opt}$ to compute the empirical FNR $\hat{\mathcal{R}}_\eta$. The CRC correction (Eq.~\ref{eq:crc}) then identifies the largest $\eta$ whose corrected risk remains within budget, yielding a valid threshold vector $\boldsymbol{\lambda}^*$ for that trajectory. Third, each valid $\boldsymbol{\lambda}^*$ is evaluated on the held-out set $\mathcal{D}_\text{valid}$, and the system retains the threshold vector achieving the lowest cardinality $\mathcal{C}(\boldsymbol{\lambda}^*)$. The resulting $\hat{\boldsymbol{\lambda}}$ is stored in a lookup table mapping $(\mathcal{T}, \alpha) \to \hat{\boldsymbol{\lambda}}$, as shown in Figure~\ref{fig:overview}~(B). This table is computed once per dataset and query topology. Like updating statistics in a traditional DBMS, recalibration is needed only when the graph structure or model weights shift enough to alter score distributions. We quantify this cost in Sec.~\ref{sec:experiments} (Table~\ref{tab:calibration_times}).

At query time, the planner retrieves the calibrated threshold vector $\hat{\boldsymbol{\lambda}}$ for the given query topology $\mathcal{T}$ and risk budget $\alpha$. The logical DAG is compiled into a physical plan of conformal gates $\bar{f}_{\lambda_{1}}^{(1)}, \ldots, \bar{f}_{\lambda_{k}}^{(k)}$, evaluated in topological order. Each gate compares its operator threshold $\lambda_i$ against the routing threshold $\delta$ to select its execution mode: retrieval-only when local graph evidence suffices, or hybrid when inference is needed to meet the recall target. Intersection and union nodes apply deterministic set operations and introduce no additional uncertainty. The final output $\hat{Y}^{(k)}_{\hat{\boldsymbol{\lambda}}}$ satisfies a marginal guarantee that the expected end-to-end recall remains above $1-\alpha$.

\subsection{Practical considerations} 
\noindent \textbf{Optimization Tractability.} Our formulation seeks the tightest threshold vector that satisfies the recall target. To achieve this, the quantile-space scalarization reduces the $k$-dimensional threshold space to a one-dimensional search over $\eta$, enabling tractable calibration with finite-sample guarantees. However, this is not guaranteed to be the global optimum over the full $k$-dimensional space, as the scalarization constrains the search to a single monotonic path. An unconstrained search could theoretically find threshold vectors yielding smaller prediction sets at the same recall level by tightening one hop more aggressively while loosening another. While advanced search techniques, such as Bayesian optimization or gradient-based searches, could theoretically explore this high-dimensional space, they face two critical hurdles. First, evaluating arbitrary threshold vectors incurs prohibitive computational costs. Second, and more fundamentally, unconstrained multi-dimensional tuning violates the strict nestedness property required by CRC. We view this as an inherent tradeoff between statistical validity and optimization power; the scalarization sacrifices precision optimality in exchange for rigorous finite-sample guarantees that hold without distributional assumptions.

\noindent \textbf{Ground Truth Requirement.} Beyond statistical assumptions, ConRAD's calibration procedure introduces a practical systems requirement: a held-out calibration set of queries with known ground-truth answers. These are necessary to compute non-conformity scores and derive the calibrated thresholds. While obtaining exhaustive ground truth in an open-world setting is inherently challenging, database administrators can bootstrap system initialization using historical query logs with verified results, manual expert validation for a small sample of queries, or synthetic query generation over densely populated, trusted subgraphs.

\noindent \textbf{Compositionality vs. End-to-End Inference.} A common alternative in neural graph reasoning is end-to-end inference, where a model predicts a multi-hop destination in a single latent step. While such models can offer superior precision by capturing global dependencies, they function as black boxes that lack the plan transparency and intermediate filtering required by modern DBMSs. ConRAD explicitly favors a compositional approach by decomposing queries into primitive link-prediction operators. This modularity ensures that each step remains interpretable and allows for frontier pruning in dense regions. Importantly, the statistical foundations of ConRAD are operator-agnostic. Future neural query optimizers could treat fused multi-hop predictors as individual physical operators within our conformal framework, trading plan granularity for predictive latency while maintaining the same formal guarantees.

\section{Experiments}

We evaluate ConRAD on multi-hop EPFO queries over incomplete knowledge graphs,
addressing four research questions:

\noindent (RQ1) Validity: Does ConRAD satisfy user-specified recall targets across diverse query topologies and graph incompleteness levels?

\noindent (RQ2) Precision: Does ConRAD's risk-constrained optimization achieve precision competitive with best-case static baselines?

\noindent (RQ3)  Efficiency and Overheads: Does the conformal gate reduce neural invocations by exploiting local graph evidence, and are the offline and online overheads tractable?

\noindent (RQ4) Robustness: Are the guarantees preserved when the
underlying neural model is swapped?






\subsection{Setup}


\noindent \textbf{Datasets.} We evaluate ConRAD on three established knowledge graph benchmarks that span distinct structural regimes (Table~\ref{tab:dataset_stats}). FB15k-237~\cite{DBLP:conf/nips/BordesUGWY13}, a densely connected subgraph of Freebase~\cite{DBLP:conf/sigmod/BollackerEPST08}, features a high average degree that stresses the framework's ability to manage frontier explosion during multi-hop traversals. NELL-995~\cite{DBLP:conf/emnlp/XiongHW17}, constructed via web-scale information extraction, introduces the noise and structural irregularity characteristic of real-world knowledge bases. Finally, YAGO3-10~\cite{DBLP:conf/cidr/MahdisoltaniBS15} serves as a high-volume scalability benchmark derived from Wikipedia and WordNet~\cite{DBLP:journals/cacm/Miller95}. To rigorously assess robustness under degraded conditions, we simulate incompleteness by removing uniformly at random 5\%, 20\%, and 40\% of edges from each dataset. These sparsity levels model realistic deployment scenarios ranging from minor data staleness (5\%) to severe structural degradation (40\%) where nearly half the relational evidence is unavailable.

\begin{table}[t]
\centering
\caption{Statistics of the KG datasets used in our evaluation.}
\label{tab:dataset_stats}
\small
\begin{tabular}{lrrr}
\hline
\textbf{Dataset} & \textbf{\# Entities} & \textbf{\# Relations} & \textbf{\# Triples} \\ \hline
FB15k-237        & 14541               & 237                   & 310116             \\
NELL-995         & 75492               & 200                   & 154213             \\
YAGO3-10         & 123182              & 37                    & 1089040           \\ \hline
\end{tabular}
\vspace{-2em}
\end{table}

\noindent \textbf{The Neural Predictor.} We employ the UltraQuery foundation model~\cite{DBLP:conf/nips/GalkinZ00Z24} as a black-box scoring function $\phi$. UltraQuery provides inductive, zero-shot predictions without task-specific fine-tuning, allowing us to isolate ConRAD's calibration efficacy from the underlying model quality. While UltraQuery is capable of end-to-end multi-hop predictions in a single inference step, we explicitly deploy it as a per-hop link predictor. This allows us to evaluate ConRAD's ability to maintain recall guarantees across composed query plans. To test robustness to predictor quality, we also evaluate UltraQuery-T, a variant with different weights and lower predictive accuracy. On the NELL-995 tail-prediction task (i.e., one-hop prediction), UltraQuery-T achieves a significantly lower MRR (48.2\% vs.\ 54.2\%) and Hits@10 (67.0\% vs.\ 70.9\%) compared to the base model. This lower-fidelity variant tests ConRAD's ability to adaptively adjust thresholds when paired with a weaker predictor, demonstrating that the recall guarantee is preserved through recalibration alone.



\noindent \textbf{Baselines.} We compare against three execution strategies. \texttt{neo4j} executes queries deterministically over the observed graph $\hat{\mathcal{G}}$, guaranteeing zero false positives but with recall bounded by graph incompleteness. \texttt{neural} executes queries hop-by-hop using UltraQuery with fixed, global thresholds ($\theta \in \{0.7, 0.8, 0.9, 0.99\}$), representing the standard workflow of manual threshold tuning without formal guarantees. \texttt{hybrid} implements the conformal gate with fixed global thresholds ($\theta \in \{0.4, 0.5, 0.6, 0.7\}$) and routing threshold $\delta=0.5$. This baseline represents the best a practitioner can achieve by combining retrieval and inference with manual threshold tuning, without offline calibration or per-operator adaptation.

\noindent \textbf{Metrics.} We report four indicators, averaged over the evaluation query set. First, the \emph{Empirical Recall} ($1 - \text{FNR}$) is defined as the fraction of ground-truth entities retrieved. Next, \emph{Precision} is defined as the fraction of returned entities that are ground-truth answers. The \emph{Abstention Rate} is the fraction of queries producing empty result sets ($|\hat{Y}^{(k)}_{\boldsymbol{\lambda}}| = 0$). Lastly, the \emph{Neural Invocations} are defined as the total number of inference operator calls per query, our primary proxy for computational overhead. Queries where the system abstains are assigned zero precision and zero recall.

\noindent \textbf{Calibration Procedure.}\label{sec:calib-exp-setup} 
We employ a split-conformal strategy, executing 5000 queries per topology over the incomplete graph $\hat{\mathcal{G}}$. The workload is partitioned into disjoint sets for CRC calibration ($|\mathcal{D}_{\text{opt}}|=2000$), cardinality-driven strategy selection ($|\mathcal{D}_{\text{valid}}|=2000$), and final evaluation (1000). For high incompleteness (40\%) or large-scale settings (YAGO3-10), the evaluation is reduced to 500 queries to manage materialization overhead. We discretize the scalarization (Section~\ref{sec:scalarization}) into a grid of 100 candidate values of $\eta$. The candidate strategy set $\Gamma$ contains five exponent vectors per query topology: a uniform strategy ($\gamma_j = 1$ for all $j$), two asymmetric strategies that prioritize early vs.\ late pruning ($\boldsymbol{\gamma} = [1.5, 1.0, 0.5]$ and $[0.5, 1.0, 1.5]$ for 3p), and two balanced variants ($\gamma_j = 0.7$ and $\gamma_j = 1.5$ for all $j$). The optimizer selects the strategy yielding the lowest cardinality on $\mathcal{D}_\text{valid}$. The calibration is performed independently per query topology.

\begin{figure}[t]
    \centering
    \includegraphics[width=0.97\linewidth]{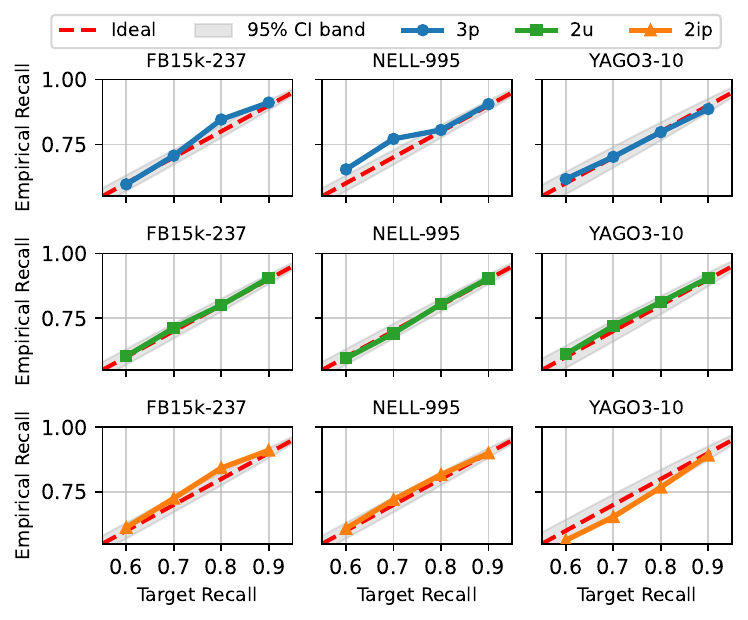}
    \caption{Validity results across query topologies, and datasets under 20\% data incompleteness level. ConRAD offers effective recall control over all the tested configurations.}
    \label{fig:validity}
    \vspace{-2em}
\end{figure}

\noindent \textbf{Frontier Management.} During calibration, loose thresholds can require materializing up to $|\mathcal{V}|^3$ candidate trajectories for \texttt{3p} queries. To maintain tractability, we restrict intermediate propagation to the union of ground-truth entities and the top-10 highest-scoring neural candidates per hop. First, because conformal thresholds are determined primarily by ground-truth scores, this pruning preserves the calibration guarantees~\cite{vovk2005algorithmic,DBLP:journals/ftml/AngelopoulosB23}. Further, valid end-to-end paths can traverse intermediate entities absent from any sub-query's ground truth, necessitating the inclusion of top-scoring neural candidates. We retain 10 per hop as a practical trade-off between path coverage and materialization cost; varying this from 5 to 20 had negligible effect on the calibrated thresholds. At runtime, calibrated thresholds are applied to all candidates without restriction.

\noindent \textbf{Experimental Scalability.} We restrict the calibration and evaluation to queries with intermediate ground-truth frontiers $|Y^{(i)}| \le 50$, due to the memory overhead of materializing $|Y^{(i)}| \times |\mathcal{V}|$ score matrices on the GPU. This is a pragmatic constraint for batch evaluation, not a framework limitation; larger frontiers are natively supported via standard techniques such as gradient accumulation or multi-GPU sharding.

\noindent \textbf{Implementation Details.} The graph store is Neo4j v5.20.0, handling all retrieval traversals and set operations (intersection, union). Inference uses a PyTorch deployment of UltraQuery on a single NVIDIA H100 NVL (96GB VRAM). The routing threshold $\delta$ is set to 0.5 in all experiments. The hash function $\mathcal{H}$ in Eq.~\ref{eq:unified-score} is a deterministic MD5 hash normalized to $[0, 1]$.

\subsection{Results}
\label{sec:experiments}
We evaluate ConRAD on three standard KG datasets across varying recall targets, query topologies, and incompleteness levels. Our experiments demonstrate four key findings:

\begin{enumerate}
    \item ConRAD strictly satisfies declarative recall targets across all the tested settings. Even under severe 40\% graph incompleteness and difficult query topologies, the empirical recall tightly tracks the target with a maximum downward deviation of only 0.0464.
    
    \item ConRAD achieves precision competitive with (and often exceeding) best-case static baselines, reaching up to 96\% on FB15k-237 \texttt{3p} queries. Unlike static methods, ConRAD eliminates silent coverage violations, maintaining low abstention rates even as incompleteness increases to 40\%.
    
    \item The conformal gate dynamically optimizes the execution by bypassing neural inference when local graph evidence suffices. At moderate recall targets in near-complete graphs (5\% incompleteness), ConRAD reduces total neural invocations to 0\%, selectively reintroducing inference only as frontiers expand or evidence degrades.

    \item ConRAD's statistical guarantees are robust to predictor quality. Substituting UltraQuery with UltraQuery-T preserves the recall target across all query topologies (maximum downward deviation: 0.017). Precision degrades proportionally to model quality, confirming that ConRAD provides formal recall guarantees without imposing accuracy requirements on the underlying predictor.
    
\end{enumerate}

\begin{figure}[t]
    \centering
    \includegraphics[width=0.97\linewidth]{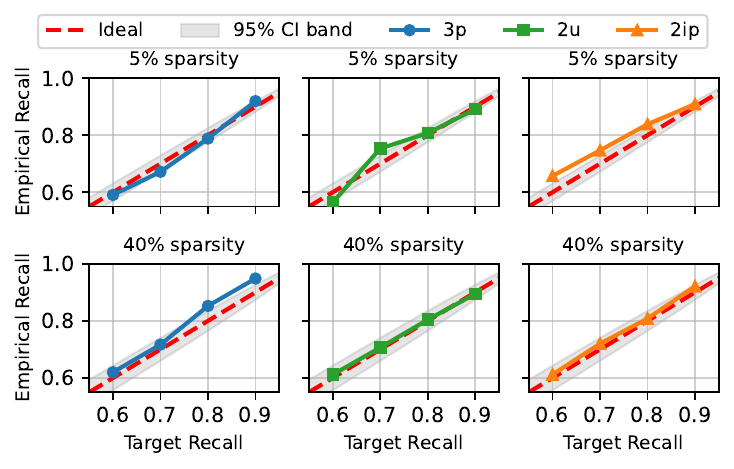}
    
    \caption{Validity results across query topologies for the FB15k-237 dataset. ConRAD offers effective recall control across all tested data incompleteness levels.}
    \label{fig:sparsity}
    \vspace{-1em}
\end{figure}

\begin{figure}[t]
    \centering
    \includegraphics[width=0.96\linewidth]{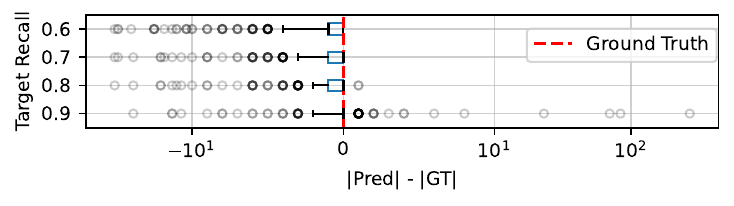}
    \caption{Per-query cardinality adaptivity ($|Pred| - |GT|$) on NELL-995 (20\% data incompleteness) for \texttt{2ip}. The tight clustering around zero at moderate recall targets (0.6-0.7) indicates high precision across most queries. The positive tail at stricter targets shows the system trades off precision for recall on difficult queries.}
    \label{fig:adaptivity}
\end{figure}

\begin{table}[t]
\centering

\caption{Maximum precision subject to empirical recall $\ge 0.60$ at 20\% incompleteness. For each baseline, we report the best-performing threshold. For ConRAD, we report the risk budget $\alpha$ yielding the highest precision at the same recall floor. \textit{Fails (X)} means the method did not reach 60\% recall, where \textit{X} is the maximum achieved.}
\label{tab:baselines}
\renewcommand{\arraystretch}{0.9}
\resizebox{\linewidth}{!}{
\begin{tabular}{@{}llcccc@{}}
\toprule
\textbf{Dataset} & \textbf{$\mathcal{T}$} & \textbf{\texttt{neo4j}} & \textbf{\texttt{neural} (best $\theta$)} & \textbf{\texttt{hybrid} (best $\theta$)} & \textbf{conrad (best $\alpha$)} \\ \midrule
\multirow{3}{*}{FB15k-237} 
 & \texttt{3p}  & Fails (0.59) & 0.82 ($\theta=0.7$) & 0.95 ($\theta=0.4$) & \textbf{0.96} ($\alpha=0.9$) \\
 & \texttt{2ip} & Fails (0.57) & 0.70 ($\theta=0.7$) & 0.92 ($\theta=0.4$) & \textbf{0.93} ($\alpha=0.9$) \\
 & \texttt{2u}  & 0.99 & 0.99 ($\theta=0.8$) & 1.00 ($\theta=0.4$) & \textbf{1.00} ($\alpha=0.9$) \\ \midrule
\multirow{3}{*}{NELL-995}  
 & \texttt{3p}  & Fails (0.36) & 0.70 ($\theta=0.7$) & \textbf{0.94} ($\theta=0.4$) & 0.93 ($\alpha=0.9$) \\
 & \texttt{2ip} & Fails (0.57) & 0.65 ($\theta=0.7$) & 0.90 ($\theta=0.4$) & \textbf{0.90} ($\alpha=0.9$) \\
 & \texttt{2u}  & 0.97 & 0.96 ($\theta=0.9$) & 0.97 ($\theta=0.4$) & \textbf{0.99} ($\alpha=0.9$) \\ \midrule
\multirow{3}{*}{YAGO3-10}  
 & \texttt{3p}  & Fails (0.56) & Fails (0.56)   & 0.85 ($\theta=0.4$) & \textbf{0.87} ($\alpha=0.8$) \\
 & \texttt{2ip} & Fails (0.53) & Fails (0.50)   & \textbf{0.85} ($\theta=0.4$) & 0.83 ($\alpha=0.9$) \\
 & \texttt{2u}  & 0.99 & 0.93 ($\theta=0.7$) & 0.99 ($\theta=0.5$) & \textbf{0.99} ($\alpha=0.8$) \\ \bottomrule
\end{tabular}
}
\end{table}

\subsubsection{Validity.} \label{sec:exps-validity} To answer RQ1, we evaluate whether ConRAD's conformal calibration translates into reliable recall control across datasets, query topologies, and incompleteness levels. Figure~\ref{fig:validity} reports empirical recall as a function of the recall target (0.6-0.9) at 20\% incompleteness across all three datasets and query topologies. ConRAD consistently satisfies the declared risk budget $\alpha$. Empirical recall tightly tracks or marginally exceeds the target, generally falling within or above the shaded 95\% confidence bands (derived from the standard error over the evaluation query set). The bands widen on YAGO3-10 and at 40\% incompleteness, where the evaluation set was reduced to 500 queries for tractability. The maximum upward deviation across all settings was 0.0711 (NELL-995, \texttt{3p}, target 0.70); the maximum downward deviation was 0.0464 (YAGO3-10, \texttt{2ip}, target 0.70). Both are consistent with finite-sample calibration and threshold grid quantization~\cite{vovk2005algorithmic, DBLP:journals/ftml/AngelopoulosB23}. Across topologies, \texttt{2u} exhibits the tightest tracking: the parallel union structure provides inherent recall redundancy, allowing the calibrator to select tighter, higher-precision thresholds. The \texttt{3p} topology is the most difficult, as errors compound sequentially across three hops. Yet, the guarantee holds robustly across all datasets. 

Figure~\ref{fig:sparsity} isolates the effect of incompleteness by evaluating ConRAD on FB15k-237 under 5\% (near-complete) and 40\% (high incompleteness) conditions. Even under these extremes, empirical recall remains tightly bounded to the target. The maximum upward deviation was 0.0573 (2ip, 5\%, target 0.60), while the maximum downward deviation was 0.0354 (2u, 5\%, target 0.60). These results confirm that ConRAD's guarantees hold across incompleteness levels. The calibration procedure is agnostic to the degree of incompleteness and requires no manual intervention or prior knowledge of the graph's structural state.

While Figures~\ref{fig:validity} and~\ref{fig:sparsity} report aggregate recall, Figure~\ref{fig:adaptivity} examines per-query behavior. For each query on NELL-995 at 20\% incompleteness, we plot the difference between predicted and ground-truth answer set sizes against the recall target. At moderate targets (0.6-0.7), a substantial fraction of queries cluster tightly around zero, indicating near-perfect precision. As the target increases, the system admits more candidates, trading precision for recall. The positive tail corresponds to queries over severely incomplete local neighborhoods. Here, retrieval yields few edges, and the inference operator produces low-confidence scores, forcing the system to admit more candidates to meet the risk budget. This confirms that ConRAD does not meet its recall targets by uniformly inflating prediction sets, but it adapts selectivity to the local difficulty of each query. We observe consistent behavior across all datasets and query topologies. We report NELL-995 \texttt{2ip} for brevity.

\begin{figure}[t]
    \centering
    \includegraphics[width=0.98\linewidth]{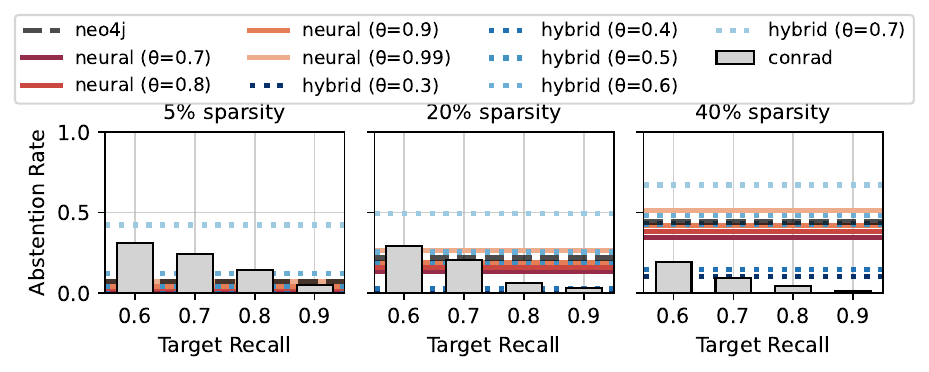}
    \caption{Query abstention rate on FB15k-237 (\texttt{3p}). Static baselines abstain on up to 60\% of queries at 40\% data incompleteness. ConRAD's abstention remains stable.}
    \label{fig:baselines-abstention}
\end{figure}

\subsubsection{Precision}

To answer RQ2, we compare ConRAD's precision against static baselines that require manual threshold tuning and provide no formal recall guarantees. Table~\ref{tab:baselines} reports the maximum precision achieved by each execution strategy subject to an empirical recall constraint of $\geq 0.6$ at 20\% data incompleteness level. For static baselines, we report the best-performing threshold from a coarse grid search. For ConRAD, we report the risk budget $\alpha$ that achieves the same recall constraint with the highest precision.


At 20\% incompleteness, deterministic retrieval (\texttt{neo4j}) fails to reach 60\% recall for both \texttt{3p} and \texttt{2ip} across all three datasets, because a missing edge at any hop prunes all downstream paths. The only topology where retrieval meets the recall floor is \texttt{2u}, where redundant parallel paths make it resilient to missing edges. All methods achieve near-perfect precision on \texttt{2u}, confirming that union queries do not stress the predictive pipeline.


Conversely, the \texttt{neural} baseline demonstrates why uncalibrated inferences are insufficient for rigorous database querying. While pure inference manages to satisfy the 60\% recall constraint on FB15k-237 and NELL-995, it suffers a severe drop in precision. The \texttt{hybrid} baseline reaches significantly higher precision ranges, confirming the value of combining retrieval with inference. However, ConRAD's system-managed recall bounding consistently matches or exceeds this \texttt{hybrid}'s exhaustive static thresholding in terms of precision. On FB15k-237, ConRAD reaches 0.96 precision on \texttt{3p} queries compared to 0.95, and 0.94 on intersection queries compared to 0.93. Even on the challenging YAGO3-10 dataset, ConRAD outperforms \texttt{hybrid} on \texttt{3p} queries (0.87 versus 0.85). Where \texttt{hybrid} leads marginally (e.g., YAGO3-10 \texttt{2ip}, 0.85 vs. 0.83), the gap is narrow. Crucially, the \texttt{hybrid} results reflect the best threshold found by a coarse grid search with no guarantee of transfer to unseen queries. ConRAD derives its thresholds automatically from the declared risk budget $\alpha$ and provides a formal recall guarantee over the query distribution.

We next examine query abstention, where the system returns an empty result set. Figure~\ref{fig:baselines-abstention} tracks the abstention rate for FB15k-237 \texttt{3p} across incompleteness levels. At 5\%, abstention is negligible for all baselines. As incompleteness increases to 40\%, the static \texttt{hybrid} baselines diverge sharply. Configurations with moderate-to-high thresholds abstain on up to 60\% of queries, as fixed thresholds that worked at lower incompleteness become too aggressive. These are silent failures, yielding zero recall on the affected instances. ConRAD maintains consistently low abstention across all incompleteness levels, staying below 30\% even at 40\% data incompleteness. Because abstention contributes zero recall to the expected risk, the CRC optimization inherently steers away from threshold vectors that produce empty results. Abstention does rise slightly at relaxed recall targets (e.g., 29.2\% at target 0.6 vs. 3.3\% at target 0.9 under 40\% incompleteness). This is a side effect of the precision objective. A generous risk budget allows tighter thresholds that produce smaller prediction sets, which can yield empty results on the hardest queries. At strict recall targets, the optimizer cannot tolerate such abstentions and selects more permissive thresholds accordingly.

\begin{figure}[t]
    \centering
    \includegraphics[width=0.98\linewidth]{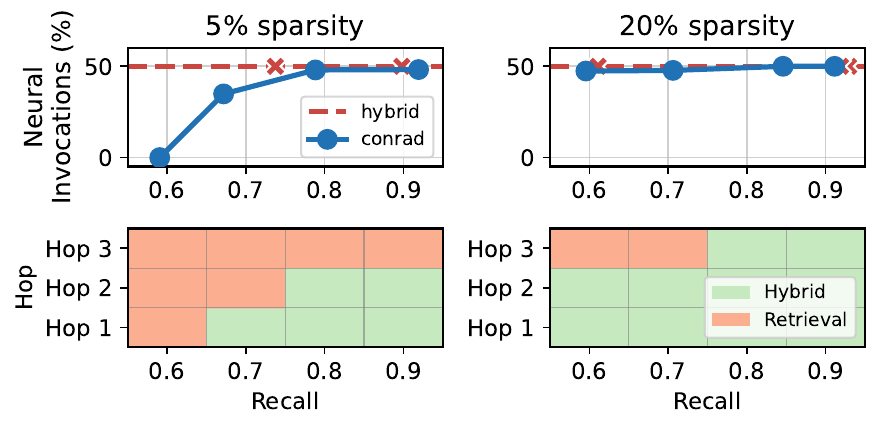}
    \caption{Conformal gate efficiency on FB15k-237 \texttt{3p}. Top: aggregate inference invocations as a fraction of total hops. Bottom: per-hop routing decisions. The gate bypasses inference at early hops where retrieval suffices and switches to hybrid mode at later hops as local evidence degrades.}
    \label{fig:neural-calls}
    \vspace{-1em}
\end{figure}

\subsubsection{Efficiency} 

To answer RQ3, we evaluate whether the conformal gate reduces neural invocations by exploiting local graph evidence. Figure~\ref{fig:neural-calls} reports the percentage of inference operator calls across the FB15k-237 \texttt{3p} pipeline. The \texttt{hybrid} baseline triggers inference at every hop regardless of local graph density or recall target. ConRAD's conformal gate bypasses inference when retrieval alone satisfies the calibrated threshold. The savings depend on both the risk budget and the incompleteness level. At a recall target of 0.6 under 5\% incompleteness, ConRAD bypasses inference entirely, reducing invocations to zero. As the target increases to 0.7 and 0.9, invocation rates rise to 34.9\% and 48.2\% respectively. The per-hop routing breakdown reveals where the savings originate. Initial hops often traverse densely connected neighborhoods where retrieval alone satisfies the threshold. At moderate recall targets, the conformal gate operates in retrieval-only mode for the first hop, reducing both the immediate inference cost and the intermediate result set size, which limits false-positive propagation to downstream operators. As execution progresses to later hops, or as incompleteness increases to 20\%, local graph evidence degrades and the conformal gate switches to hybrid mode to meet the recall target. Despite bypassing inference at early hops, the aggregate invocation rate is dominated by later hops, which process a larger candidate set due to frontier expansion. At 20\% incompleteness, ConRAD's invocation rate plateaus at 50\% for recall targets of both 0.8 and 0.9, reflecting this effect. The conformal gate thus concentrates inference where retrieval evidence is insufficient, reducing total model calls while satisfying the recall guarantee.

\begin{table}[t]
\caption{Offline calibration wall-clock time (sec). Each cell reports the \textbf{Total Time} (\textit{$\lambda$ Optimization only}), where Total = Scores Computation + $\lambda$ Optimization. The scores computation dominates the overhead.}
\label{tab:calibration_times}
\vspace{-1em}
\centering
\renewcommand{\arraystretch}{0.9}
\small
\begin{tabular}{lrrr}
\toprule
Dataset & \multicolumn{1}{c}{\texttt{3p}} & \multicolumn{1}{c}{\texttt{2u}} & \multicolumn{1}{c}{\texttt{2ip}} \\
\midrule
FB15k-237 & 1302.2 (344)    & 233.1 (4.5)  & 750.2 (22.1)  \\
NELL-995  & 2095.0 (381.6)  & 603.3 (7.4)  & 6142.5 (81.4) \\
YAGO3-10  & 3711.3 (927.2)  & 1327.9 (7.7) & 4476.7 (91.2) \\
\bottomrule
\end{tabular}
\vspace{-1em}
\end{table}


\subsubsection{Overheads} Table~\ref{tab:calibration_times} reports wall-clock calibration time across datasets and query topologies at 20\% incompleteness. The total cost remains tractable, reaching at most approximately 1.5 hours for the largest setting (YAGO3-10 \texttt{3p}). A breakdown reveals that score computation dominates the overhead. On YAGO3-10 \texttt{2u}, it accounts for over 99\% of the total time, while the threshold optimization requires only 7.7 seconds. This phase consists entirely of evaluating the neural model over graph triples, a process that can be decoupled from the optimization itself. In a production environment, the system can gather non-conformity scores asynchronously by piggybacking on online query execution. Once a sufficient sample is collected, the optimizer only needs to trigger the lightweight threshold search to produce updated calibration parameters.

Lastly, the runtime overhead introduced by ConRAD is negligible. The gate's routing logic requires only computing the unified score via a lightweight, deterministic hash function for retrieving facts and a scalar projection for inference scores, followed by a single floating-point comparison against the calibrated threshold per candidate tuple. This ensures that the formal recall guarantees add only constant-time overhead per candidate triple.

\begin{figure}[t]
    \centering
    \includegraphics[width=0.98\linewidth]{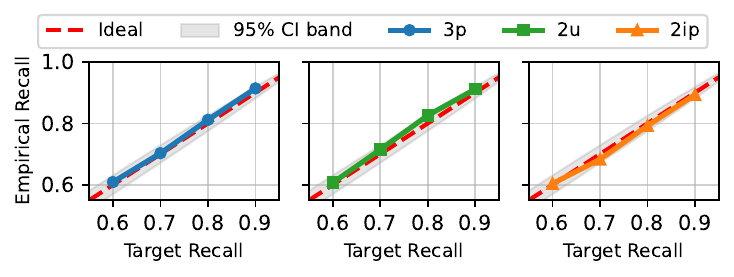}
    \caption{Robustness to predictor quality (RQ4): ConRAD with UltraQuery-T on NELL-995 (20\% sparsity). Despite a weaker underlying model, the recall guarantee is preserved across all three query templates, requiring only offline recalibration. Maximum precision reaches 0.70 (\texttt{3p}), 0.66 (\texttt{2ip}), and 0.96 (\texttt{2u}), compared to 0.93, 0.9, and 0.99 for UltraQuery under identical conditions (Table~\ref{tab:baselines}).}
    \label{fig:ultraqueryt}
    \vspace{-2em}
\end{figure}

\subsubsection{Robustness} To answer RQ4, we evaluate whether ConRAD's statistical guarantees transfer when the underlying neural model is replaced. We substitute UltraQuery with UltraQuery-T, a less accurate variant of the same architecture, and recalibrate on NELL-995 at 20\% sparsity. No other system component is modified.  Figure~\ref{fig:ultraqueryt} shows that the recall guarantee is preserved across all three query templates, with a maximum upward deviation of 0.0259 (target 0.80, \texttt{2u}) and a maximum downward deviation of 0.0169 (target 0.70, \texttt{2ip}). Both are well within the finite-sample variance, as noted also in Sec.~\ref{sec:exps-validity}. As expected, precision degrades relative to UltraQuery, reflecting the weaker model's noisier score distribution. Maximum precision across templates reaches 0.70 (3p), 0.96 (2u), and 0.66 (2ip), compared to 0.93, 0.9, and 0.99 for UltraQuery under identical conditions (Table~\ref{tab:baselines}). This confirms that ConRAD cannot compensate for a fundamentally less accurate predictor, but the safety contract remains intact. Practitioners can therefore upgrade or replace the scoring model and restore guarantees through recalibration alone, without modifying the framework or query plans.

\section{Discussion}


\noindent \textbf{Marginal Guarantees and Exchangeability.} ConRAD provides marginal guarantees, ensuring the recall target is satisfied on average across the query distribution rather than per individual query instance. Achieving strict conditional coverage is provably impossible without restrictive distributional assumptions or infinite calibration data~\cite{foygel2021limits}. While this trade-off enables distribution-free guarantees, it fundamentally relies on the exchangeability of calibration and test queries. The premise of exchangeability over non-i.i.d.\ graph structures has been formally established for Graph Neural Networks~\cite{huang2023uncertainty} and specifically for link prediction when queries are sampled uniformly~\cite{DBLP:conf/kdd/ZhaoK024}. However, as we have empirically demonstrated in prior work~\cite{DBLP:journals/pvldb/HorchidanZBC25}, conformal guarantees can degrade under severe distribution shifts (e.g., sudden query skew toward unseen, high-variance relations), which is a limitation that applies equally to ConRAD's predictive pipelines. We discuss potential mitigation strategies, including adaptive conformal methods designed for distribution shift, in Sec.~\ref{sec:future}.


\noindent \textbf{Zero-Shot Transfer Constraints.} ConRAD requires a representative calibration workload for each deployment setting and does not currently support zero-shot transfer to novel query topologies or out-of-distribution workloads without recalibration. When deploying on a new query template without calibration data, the system must invoke a conservative fallback, such as the union bound alternative described in Sec.~\ref{sec:method}. However, this limitation does not preclude the modular composition of pre-calibrated branches. If disjoint query branches have already been independently calibrated, they can be safely composed into novel query topologies via deterministic set operators. In such cases, the system could apply structural risk composition rules to primitive sub-plans rather than at the individual operator level, potentially avoiding the severe precision loss of a pipeline-wide fallback. We leave formal analysis of this compositional strategy to future work.


\section{Future work} \label{sec:future}

\noindent \textbf{Compute-Driven Optimization Objectives.} While this paper optimizes for end-to-end precision, the cost function introduced in Sec.~\ref{sec:problem} can alternatively target computational effort. Future work will investigate a compute-driven objective that minimizes query latency and I/O cost. This inherently incentivizes tight thresholds early in the plan to aggressively prune intermediate results, a stochastic analogue to classical predicate pushdown. However, this introduces a tension with end-to-end precision: tight early thresholds force downstream operators to adopt more permissive thresholds to satisfy the risk budget $\alpha$, admitting more false positives into the final output. Navigating this trade-off offers a rich space for future query optimization research.

\noindent \textbf{Continuous Online Recalibration.} To address the static nature of offline calibration and mitigate workload drift, we plan to explore continuous online recalibration. As shown in Sec.~\ref{sec:experiments}, score computation dominates the offline phase, so the framework could naturally support asynchronous recalibration by piggybacking score collection onto online query execution and triggering the lightweight threshold optimization incrementally. Furthermore, integrating adaptive conformal methods designed for distribution shift, such as Weighted Conformal Prediction~\cite{barber_conformal_2023, DBLP:conf/nips/TibshiraniBCR19} and Adaptive Conformal Inference~\cite{gibbs_adaptive_2021}, could allow the framework to adapt to evolving graph structures and shifting query distributions.

\noindent \textbf{Generalization to Broader AI Systems.} Finally, the ConRAD methodology applies to any predictive pipeline that forms a DAG of operators where stochastic components produce monotonic prediction sets with respect to a threshold parameter. Empirical validation of the framework beyond KGs queries is an important direction for future work. We plan to extend these finite-sample recall guarantees to other composite architectures, such as RAG pipelines, learned join operators, and complex multi-stage retrieval systems.
\section{Related work}

\noindent\textbf{AI-Integrated Database Systems.} Modern data management systems increasingly embed learned components to replace or augment traditional operations, from learned indexes~\cite{DBLP:conf/sigmod/KraskaBCDP18} and query optimizers~\cite{DBLP:journals/sigmod/MarcusNMTAK22} to systems that execute queries over incomplete data. Prior works optimize queries involving expensive ML-based predicates through cascade routing and cost-quality trade-offs (e.g., NoScope~\cite{kang2017noscope}, BlazeIt~\cite{DBLP:journals/pvldb/KangBZ19}, VIVA~\cite{DBLP:journals/pvldb/RomeroHPKZK22}, EVAPORATE~\cite{DBLP:journals/pvldb/AroraYENHTR23}). Raven~\cite{DBLP:conf/sigmod/ParkSBSIK22, DBLP:conf/cidr/KaranasosIPSPPX20} presents a unified optimization framework for prediction queries that compose data processing operators with ML inference in a single plan. Neural Graph Databases~\cite{DBLP:journals/tmlr/Ren0ZLC24, DBLP:conf/log/BestaISODPCH22, DBLP:conf/vldb/Horchidan23, DBLP:conf/esws/HorchidanC23} compose exact graph traversals with neural link prediction, CAESURA~\cite{DBLP:conf/cidr/UrbanB24} integrates ML inference into relational SQL plans, InferDB~\cite{DBLP:journals/pvldb/SalazarDiazGR24} accelerates in-database inference by approximating ML pipelines through index lookups, and compound AI frameworks such as DSPy~\cite{khattab2023dspy}, LOTUS~\cite{patel2024semantic}, and Palimpzest~\cite{DBLP:conf/cidr/LiuRC0CCFK0SV25} orchestrate multi-model workflows over structured and unstructured data. While these frameworks improve expressiveness, orchestration, or latency, they fundamentally lack formal, distribution-free guarantees for end-to-end correctness.


\noindent\textbf{Approximate Query Processing and Uncertainty in Databases.} Managing uncertainty natively within the query engine is a foundational database problem. Traditional Approximate Query Processing systems such as BlinkDB~\cite{agarwal2013blinkdb}, VerdictDB~\cite{park2018verdictdb}, and online aggregation~\cite{hellerstein1997online} provide statistical error bounds, but these quantify sampling variance over completely observed data rather than uncertainty from learned models. A cornerstone of database research, Probabilistic Databases (e.g., MayBMS~\cite{DBLP:conf/icde/AntovaKO07a}, Trio~\cite{aggarwal2009trio}, and Dalvi and Suciu~\cite{DBLP:conf/vldb/DalviS04}) rigorously propagate tuple-level uncertainty through query evaluation. While foundational, these systems typically model a closed world of explicitly annotated probabilities. In contrast, ConRAD addresses open-world incompleteness, where unobserved facts must be dynamically recovered and bounded using uncalibrated neural scoring. To establish formal guarantees over such black-box models, the database community has recently begun adopting Conformal Prediction. For example, dbET~\cite{li2023dbet} uses CP to bound execution time distributions for cost-based plan selection, \citet{DBLP:journals/pvldb/LiuGS25} apply it to verify learned query optimizers with latency bounds, and ConANN~\cite{DBLP:journals/pvldb/HorchidanZBC25} provides recall guarantees for approximate kNN search. While these methods successfully manage model-driven uncertainty in databases, they operate exclusively at the single-operator level or target system performance metrics. ConRAD elevates conformal calibration from isolated operators to the multi-hop pipeline level.

\noindent\textbf{Uncertainty in Knowledge Graph Reasoning.} A rich landscape of neural models has been developed for KG reasoning. Link prediction models such as TransE~\cite{bordes2013translating}, ComplEx~\cite{trouillon2016complex}, and RotatE~\cite{sun2019rotate} learn expressive scoring functions for individual triples. Concurrently, neural query embedding frameworks (e.g., Query2Box~\cite{ren2020query2box}, BetaE~\cite{ren2020beta}, NQE~\cite{DBLP:conf/aaai/LuoEYZGYTLW23}, GNN-QE~\cite{zhu2022neural}) and foundation models (ULTRA~\cite{galkin2023towards}, UltraQuery~\cite{DBLP:conf/nips/GalkinZ00Z24}) extend this to complex multi-hop logical queries. However, these models optimize for pointwise accuracy and produce uncalibrated confidence scores. Standard post-hoc calibration techniques (e.g., Platt scaling~\cite{platt1999probabilistic}, temperature scaling~\cite{guo2017calibration}) improve point-wise score reliability but provide no finite-sample guarantees for set-valued outputs. To address this, recent work has applied conformal prediction to Graph Neural Networks (CF-GNN~\cite{huang2023uncertainty}) and individual KG link predictors~\cite{zhu2025conformalized, DBLP:conf/kdd/ZhaoK024}, while Bayesian approaches~\cite{DBLP:conf/aaai/ChenCSSZ19} explicitly model predictive uncertainty in evolving graphs. However, these methods calibrate individual operators in isolation. ConRAD extends rigorous uncertainty quantification beyond isolated predictors, delivering dynamically routed, end-to-end recall guarantees for composed query pipelines.

\section{Conclusions}
ConRAD addresses a fundamental gap in modern data systems by transitioning reliability from a manually tuned hyperparameter to a declarative system constraint. By extending Conformal Risk Control to multi-operator query topologies, we provide the first framework capable of delivering finite-sample recall guarantees for complex queries over incomplete knowledge graphs. Our evaluation demonstrates that ConRAD strictly preserves user-specified recall targets across varying query topologies and incompleteness levels, achieving precision competitive with best-case static baselines. As the industry moves toward composing stochastic AI primitives into broader data processing stacks, ConRAD provides a principled methodology for maintaining the database correctness contract without sacrificing the predictive power of learned models.




\bibliographystyle{ACM-Reference-Format}
\bibliography{sample}

@String{Computing = "Computing" }

@String{Computer = "{IEEE} Computer" }

@String{Springer = "Springer-Verlag" }

@article{DBLP:journals/pvldb/SalazarDiazGR24,
  author       = {Ricardo Salazar{-}D{\'{\i}}az and
                  Boris Glavic and
                  Tilmann Rabl},
  title        = {InferDB: In-Database Machine Learning Inference Using Indexes},
  journal      = {Proc. {VLDB} Endow.},
  volume       = {17},
  number       = {8},
  pages        = {1830--1842},
  year         = {2024},
  url          = {https://www.vldb.org/pvldb/vol17/p1830-salazar-diaz.pdf},
  doi          = {10.14778/3659437.3659441},
  bibsource    = {dblp computer science bibliography, https://dblp.org}
}

@article{DBLP:journals/pvldb/RomeroHPKZK22,
  author       = {Francisco Romero and
                  Johann Hauswald and
                  Aditi Partap and
                  Daniel Kang and
                  Matei Zaharia and
                  Christos Kozyrakis},
  title        = {Optimizing Video Analytics with Declarative Model Relationships},
  journal      = {Proc. {VLDB} Endow.},
  volume       = {16},
  number       = {3},
  pages        = {447--460},
  year         = {2022},
  url          = {https://www.vldb.org/pvldb/vol16/p447-romero.pdf},
  doi          = {10.14778/3570690.3570695},
  bibsource    = {dblp computer science bibliography, https://dblp.org}
}

@ArtifactSoftware{R,
    title = {R: A Language and Environment for Statistical Computing},
    author = {{R Core Team}},
    organization = {R Foundation for Statistical Computing},
    address = {Vienna, Austria},
    year = {2019},
    url = {https://www.R-project.org/},
}

@article{huang2023uncertainty,
  title={Uncertainty quantification over graph with conformalized graph neural networks},
  author={Huang, Kexin and Jin, Ying and Candes, Emmanuel and Leskovec, Jure},
  journal={Advances in Neural Information Processing Systems},
  volume={36},
  pages={26699--26721},
  year={2023}
}

@inproceedings{DBLP:conf/vldb/Horchidan23,
  author       = {Sonia Horchidan},
  title        = {Query Optimization for Inference-Based Graph Databases},
  booktitle    = {PhD@VLDB},
  series       = {{CEUR} Workshop Proceedings},
  volume       = {3452},
  pages        = {33--36},
  publisher    = {CEUR-WS.org},
  year         = {2023}
}

@inproceedings{DBLP:conf/kdd/ZhaoK024,
  author       = {Tianyi Zhao and
                  Jian Kang and
                  Lu Cheng},
  title        = {Conformalized Link Prediction on Graph Neural Networks},
  booktitle    = {{KDD}},
  pages        = {4490--4499},
  publisher    = {{ACM}},
  year         = {2024}
}

@inproceedings{DBLP:conf/aaai/LuoEYZGYTLW23,
  author       = {Haoran Luo and
                  Haihong E and
                  Yuhao Yang and
                  Gengxian Zhou and
                  Yikai Guo and
                  Tianyu Yao and
                  Zichen Tang and
                  Xueyuan Lin and
                  Kaiyang Wan},
  title        = {{NQE:} N-ary Query Embedding for Complex Query Answering over Hyper-Relational
                  Knowledge Graphs},
  booktitle    = {{AAAI}},
  pages        = {4543--4551},
  publisher    = {{AAAI} Press},
  year         = {2023}
}

@article{DBLP:journals/cacm/Miller95,
  author       = {George A. Miller},
  title        = {WordNet: {A} Lexical Database for English},
  journal      = {Commun. {ACM}},
  volume       = {38},
  number       = {11},
  pages        = {39--41},
  year         = {1995}
}

@article{DBLP:journals/pvldb/HorchidanZBC25,
  author       = {Sonia Horchidan and
                  Fabian Zeiher and
                  Henrik Bostr{\"{o}}m and
                  Paris Carbone},
  title        = {ConANN: Conformal Approximate Nearest Neighbor Search},
  journal      = {Proc. {VLDB} Endow.},
  volume       = {19},
  number       = {1},
  pages        = {29--42},
  year         = {2025}
}

@inproceedings{DBLP:conf/cidr/MahdisoltaniBS15,
  author       = {Farzaneh Mahdisoltani and
                  Joanna Biega and
                  Fabian M. Suchanek},
  title        = {{YAGO3:} {A} Knowledge Base from Multilingual Wikipedias},
  booktitle    = {{CIDR}},
  publisher    = {www.cidrdb.org},
  year         = {2015}
}

@inproceedings{DBLP:conf/sigmod/BollackerEPST08,
  author       = {Kurt D. Bollacker and
                  Colin Evans and
                  Praveen K. Paritosh and
                  Tim Sturge and
                  Jamie Taylor},
  title        = {Freebase: a collaboratively created graph database for structuring
                  human knowledge},
  booktitle    = {{SIGMOD} Conference},
  pages        = {1247--1250},
  publisher    = {{ACM}},
  year         = {2008}
}

@inproceedings{DBLP:conf/nips/BordesUGWY13,
  author       = {Antoine Bordes and
                  Nicolas Usunier and
                  Alberto Garc{\'{\i}}a{-}Dur{\'{a}}n and
                  Jason Weston and
                  Oksana Yakhnenko},
  title        = {Translating Embeddings for Modeling Multi-relational Data},
  booktitle    = {{NIPS}},
  pages        = {2787--2795},
  year         = {2013}
}

@book{abiteboul1995foundations,
  title={Foundations of databases},
  author={Abiteboul, Serge and Hull, Richard and Vianu, Victor},
  volume={8},
  year={1995},
  publisher={Addison-Wesley Reading}
}

@article{DBLP:journals/ftml/AngelopoulosB23,
  author       = {Anastasios N. Angelopoulos and
                  Stephen Bates},
  title        = {Conformal Prediction: {A} Gentle Introduction},
  journal      = {Found. Trends Mach. Learn.},
  volume       = {16},
  number       = {4},
  pages        = {494--591},
  year         = {2023}
}

@inproceedings{DBLP:conf/iclr/AngelopoulosBFL24,
  author       = {Anastasios Nikolas Angelopoulos and
                  Stephen Bates and
                  Adam Fisch and
                  Lihua Lei and
                  Tal Schuster},
  title        = {Conformal Risk Control},
  booktitle    = {{ICLR}},
  publisher    = {OpenReview.net},
  year         = {2024}
}

@article{DBLP:journals/jmlr/ShaferV08,
  author       = {Glenn Shafer and
                  Vladimir Vovk},
  title        = {A Tutorial on Conformal Prediction},
  journal      = {J. Mach. Learn. Res.},
  volume       = {9},
  pages        = {371--421},
  year         = {2008}
}

@inproceedings{DBLP:conf/esws/HorchidanC23,
  author       = {Sonia Horchidan and
                  Paris Carbone},
  title        = {{ORB:} Empowering Graph Queries through Inference},
  booktitle    = {{ESWC} Workshops},
  series       = {{CEUR} Workshop Proceedings},
  volume       = {3443},
  publisher    = {CEUR-WS.org},
  year         = {2023}
}

@inproceedings{DBLP:conf/log/BestaISODPCH22,
  author       = {Maciej Besta and
                  Patrick Iff and
                  Florian Scheidl and
                  Kazuki Osawa and
                  Nikoli Dryden and
                  Michal Podstawski and
                  Tiancheng Chen and
                  Torsten Hoefler},
  title        = {Neural Graph Databases},
  booktitle    = {LoG},
  series       = {Proceedings of Machine Learning Research},
  volume       = {198},
  pages        = {31},
  publisher    = {{PMLR}},
  year         = {2022}
}

@article{DBLP:journals/tmlr/Ren0ZLC24,
  author       = {Hongyu Ren and
                  Mikhail Galkin and
                  Zhaocheng Zhu and
                  Jure Leskovec and
                  Michael Cochez},
  title        = {Neural Graph Reasoning: {A} Survey on Complex Logical Query Answering},
  journal      = {Trans. Mach. Learn. Res.},
  volume       = {2024},
  year         = {2024}
}

@inproceedings{zhu2022neural,
  title={Neural-symbolic models for logical queries on knowledge graphs},
  author={Zhu, Zhaocheng and Galkin, Mikhail and Zhang, Zuobai and Tang, Jian},
  booktitle={International conference on machine learning},
  pages={27454--27478},
  year={2022},
  organization={PMLR}
}

@inproceedings{DBLP:conf/nips/GalkinZ00Z24,
  author       = {Michael Galkin and
                  Jincheng Zhou and
                  Bruno Ribeiro and
                  Jian Tang and
                  Zhaocheng Zhu},
  title        = {A Foundation Model for Zero-shot Logical Query Reasoning},
  booktitle    = {NeurIPS},
  year         = {2024}
}

@article{hogan2021knowledge,
  title={Knowledge graphs},
  author={Hogan, Aidan and Blomqvist, Eva and Cochez, Michael and d’Amato, Claudia and Melo, Gerard De and Gutierrez, Claudio and Kirrane, Sabrina and Gayo, Jos{\'e} Emilio Labra and Navigli, Roberto and Neumaier, Sebastian and others},
  journal={ACM Computing Surveys (Csur)},
  volume={54},
  number={4},
  pages={1--37},
  year={2021},
  publisher={ACM New York, NY, USA}
}

@inproceedings{galkin2023towards,
  author       = {Mikhail Galkin and
                  Xinyu Yuan and
                  Hesham Mostafa and
                  Jian Tang and
                  Zhaocheng Zhu},
  title        = {Towards Foundation Models for Knowledge Graph Reasoning},
  booktitle    = {{ICLR}},
  publisher    = {OpenReview.net},
  year         = {2024}
}

@inproceedings{zhu2025conformalized,
  title={Conformalized answer set prediction for knowledge graph embedding},
  author={Zhu, Yuqicheng and Potyka, Nico and Pan, Jiarong and Xiong, Bo and He, Yunjie and Kharlamov, Evgeny and Staab, Steffen},
  booktitle={Proceedings of the 2025 Conference of the Nations of the Americas Chapter of the Association for Computational Linguistics: Human Language Technologies (Volume 1: Long Papers)},
  pages={731--750},
  year={2025}
}

@inproceedings{agarwal2013blinkdb,
  title={BlinkDB: queries with bounded errors and bounded response times on very large data},
  author={Agarwal, Sameer and Mozafari, Barzan and Panda, Aurojit and Milner, Henry and Madden, Samuel and Stoica, Ion},
  booktitle={Proceedings of the 8th ACM European conference on computer systems},
  pages={29--42},
  year={2013}
}

@inproceedings{ren2020query2box,
  title={Query2box: Reasoning Over Knowledge Graphs In Vector Space Using Box Embeddings},
  author={Ren, H and Hu, W and Leskovec, J},
  booktitle={International Conference on Learning Representations (ICLR)},
  year={2020}
}

@article{ren2020beta,
  title={Beta embeddings for multi-hop logical reasoning in knowledge graphs},
  author={Ren, Hongyu and Leskovec, Jure},
  journal={Advances in Neural Information Processing Systems},
  volume={33},
  pages={19716--19726},
  year={2020}
}

@inproceedings{park2018verdictdb,
  title={Verdictdb: Universalizing approximate query processing},
  author={Park, Yongjoo and Mozafari, Barzan and Sorenson, Joseph and Wang, Junhao},
  booktitle={Proceedings of the 2018 International Conference on Management of Data},
  pages={1461--1476},
  year={2018}
}

@inproceedings{hellerstein1997online,
  title={Online aggregation},
  author={Hellerstein, Joseph M and Haas, Peter J and Wang, Helen J},
  booktitle={Proceedings of the 1997 ACM SIGMOD international conference on Management of data},
  pages={171--182},
  year={1997}
}

@inproceedings{DBLP:conf/emnlp/XiongHW17,
  author       = {Wenhan Xiong and
                  Thien Hoang and
                  William Yang Wang},
  title        = {DeepPath: {A} Reinforcement Learning Method for Knowledge Graph Reasoning},
  booktitle    = {{EMNLP}},
  pages        = {564--573},
  publisher    = {Association for Computational Linguistics},
  year         = {2017}
}

@inproceedings{DBLP:conf/cidr/UrbanB24,
  author       = {Matthias Urban and
                  Carsten Binnig},
  title        = {{CAESURA:} Language Models as Multi-Modal Query Planners},
  booktitle    = {14th Conference on Innovative Data Systems Research, {CIDR} 2024,
                  Chaminade, HI, USA, January 14-17, 2024},
  publisher    = {www.cidrdb.org},
  year         = {2024},
  url          = {https://vldb.org/cidrdb/2024/caesura-language-models-as-multi-modal-query-planners.html},
  bibsource    = {dblp computer science bibliography, https://dblp.org}
}

@inproceedings{DBLP:conf/cidr/LiuRC0CCFK0SV25,
  author       = {Chunwei Liu and
                  Matthew Russo and
                  Michael J. Cafarella and
                  Lei Cao and
                  Peter Baile Chen and
                  Zui Chen and
                  Michael J. Franklin and
                  Tim Kraska and
                  Samuel Madden and
                  Rana Shahout and
                  Gerardo Vitagliano},
  title        = {Palimpzest: Optimizing AI-Powered Analytics with Declarative Query
                  Processing},
  booktitle    = {{CIDR}},
  publisher    = {www.cidrdb.org},
  year         = {2025}
}

@inproceedings{DBLP:conf/cidr/KaranasosIPSPPX20,
  author       = {Konstantinos Karanasos and
                  Matteo Interlandi and
                  Fotis Psallidas and
                  Rathijit Sen and
                  Kwanghyun Park and
                  Ivan Popivanov and
                  Doris Xin and
                  Supun Nakandala and
                  Subru Krishnan and
                  Markus Weimer and
                  Yuan Yu and
                  Raghu Ramakrishnan and
                  Carlo Curino},
  title        = {Extending Relational Query Processing with {ML} Inference},
  booktitle    = {10th Conference on Innovative Data Systems Research, {CIDR} 2020,
                  Amsterdam, The Netherlands, January 12-15, 2020, Online Proceedings},
  publisher    = {www.cidrdb.org},
  year         = {2020},
  url          = {https://vldb.org/cidrdb/2020/extending-relational-query-processing-with-ml-inference.html},
  bibsource    = {dblp computer science bibliography, https://dblp.org}
}

@inproceedings{DBLP:conf/sigmod/ParkSBSIK22,
  author       = {Kwanghyun Park and
                  Karla Saur and
                  Dalitso Banda and
                  Rathijit Sen and
                  Matteo Interlandi and
                  Konstantinos Karanasos},
  editor       = {Zachary G. Ives and
                  Angela Bonifati and
                  Amr El Abbadi},
  title        = {End-to-end Optimization of Machine Learning Prediction Queries},
  booktitle    = {{SIGMOD} '22: International Conference on Management of Data, Philadelphia,
                  PA, USA, June 12 - 17, 2022},
  pages        = {587--601},
  publisher    = {{ACM}},
  year         = {2022},
  url          = {https://doi.org/10.1145/3514221.3526141},
  doi          = {10.1145/3514221.3526141},
  bibsource    = {dblp computer science bibliography, https://dblp.org}
}

@article{DBLP:journals/pvldb/AroraYENHTR23,
  author       = {Simran Arora and
                  Brandon Yang and
                  Sabri Eyuboglu and
                  Avanika Narayan and
                  Andrew Hojel and
                  Immanuel Trummer and
                  Christopher R{\'{e}}},
  title        = {Language Models Enable Simple Systems for Generating Structured Views
                  of Heterogeneous Data Lakes},
  journal      = {Proc. {VLDB} Endow.},
  volume       = {17},
  number       = {2},
  pages        = {92--105},
  year         = {2023},
  url          = {https://www.vldb.org/pvldb/vol17/p92-arora.pdf},
  doi          = {10.14778/3626292.3626294},
  bibsource    = {dblp computer science bibliography, https://dblp.org}
}

@article{DBLP:journals/sigmod/MarcusNMTAK22,
  author       = {Ryan Marcus and
                  Parimarjan Negi and
                  Hongzi Mao and
                  Nesime Tatbul and
                  Mohammad Alizadeh and
                  Tim Kraska},
  title        = {Bao: Making Learned Query Optimization Practical},
  journal      = {{SIGMOD} Rec.},
  volume       = {51},
  number       = {1},
  pages        = {6--13},
  year         = {2022},
  url          = {https://doi.org/10.1145/3542700.3542703},
  doi          = {10.1145/3542700.3542703},
  bibsource    = {dblp computer science bibliography, https://dblp.org}
}

@inproceedings{DBLP:conf/sigmod/KraskaBCDP18,
  author       = {Tim Kraska and
                  Alex Beutel and
                  Ed H. Chi and
                  Jeffrey Dean and
                  Neoklis Polyzotis},
  editor       = {Gautam Das and
                  Christopher M. Jermaine and
                  Philip A. Bernstein},
  title        = {The Case for Learned Index Structures},
  booktitle    = {Proceedings of the 2018 International Conference on Management of
                  Data, {SIGMOD} Conference 2018, Houston, TX, USA, June 10-15, 2018},
  pages        = {489--504},
  publisher    = {{ACM}},
  year         = {2018},
  url          = {https://doi.org/10.1145/3183713.3196909},
  doi          = {10.1145/3183713.3196909},
  bibsource    = {dblp computer science bibliography, https://dblp.org}
}

@book{vovk2005algorithmic,
  title={Algorithmic learning in a random world},
  author={Vovk, Vladimir and Gammerman, Alexander and Shafer, Glenn},
  year={2005},
  publisher={Springer}
}

@inproceedings{DBLP:conf/icde/AntovaKO07a,
  author       = {Lyublena Antova and
                  Christoph Koch and
                  Dan Olteanu},
  title        = {MayBMS: Managing Incomplete Information with Probabilistic World-Set
                  Decompositions},
  booktitle    = {{ICDE}},
  pages        = {1479--1480},
  publisher    = {{IEEE} Computer Society},
  year         = {2007}
}

@inproceedings{DBLP:conf/vldb/DalviS04,
  author       = {Nilesh N. Dalvi and
                  Dan Suciu},
  title        = {Efficient Query Evaluation on Probabilistic Databases},
  booktitle    = {{VLDB}},
  pages        = {864--875},
  publisher    = {Morgan Kaufmann},
  year         = {2004}
}

@incollection{aggarwal2009trio,
  title={Trio a system for data uncertainty and lineage},
  author={Aggarwal, Charu C},
  booktitle={Managing and Mining Uncertain Data},
  pages={1--35},
  year={2009},
  publisher={Springer}
}

@article{lewis2020retrieval,
  title={Retrieval-augmented generation for knowledge-intensive nlp tasks},
  author={Lewis, Patrick and Perez, Ethan and Piktus, Aleksandra and Petroni, Fabio and Karpukhin, Vladimir and Goyal, Naman and K{\"u}ttler, Heinrich and Lewis, Mike and Yih, Wen-tau and Rockt{\"a}schel, Tim and others},
  journal={Advances in neural information processing systems},
  volume={33},
  pages={9459--9474},
  year={2020}
}

@article{chandak2023building,
  title={Building a knowledge graph to enable precision medicine},
  author={Chandak, Payal and Huang, Kexin and Zitnik, Marinka},
  journal={Scientific data},
  volume={10},
  number={1},
  pages={67},
  year={2023},
  publisher={Nature Publishing Group UK London}
}

@article{vayena2018machine,
    author="Effy, Vayena and Alessandro, Blasimme and I. Glenn Cohen",
    title="Machine learning in medicine: Addressing ethical challenges",
    journal="PLOS Medicine",
    ISSN="1549-1676",
    publisher="Public Library of Science (PLoS)",
    year="2018",
    month="11",
    volume="15",
    number="11",
    pages="e1002689",
    DOI="10.1371/journal.pmed.1002689",
    URL="https://cir.nii.ac.jp/crid/1360861291477987968"
}

@inproceedings{
  tabacof2020probability,
  title={Probability Calibration for Knowledge Graph Embedding Models},
  author={Tabacof, Pedro and Costabello, Luca},
  booktitle={International Conference on Learning Representations (ICLR)},
  year={2020},
  url={https://openreview.net/forum?id=S1g8K1BFwS}
}

@inproceedings{guo2017calibration,
  title={On calibration of modern neural networks},
  author={Guo, Chuan and Pleiss, Geoff and Sun, Yu and Weinberger, Kilian Q},
  booktitle={International conference on machine learning},
  pages={1321--1330},
  year={2017},
  organization={PMLR}
}

@inproceedings{DBLP:conf/aaai/ChenCSSZ19,
  author       = {Xuelu Chen and
                  Muhao Chen and
                  Weijia Shi and
                  Yizhou Sun and
                  Carlo Zaniolo},
  title        = {Embedding Uncertain Knowledge Graphs},
  booktitle    = {{AAAI}},
  pages        = {3363--3370},
  publisher    = {{AAAI} Press},
  year         = {2019}
}

@article{platt1999probabilistic,
  title={Probabilistic outputs for support vector machines and comparisons to regularized likelihood methods},
  author={Platt, John and others},
  journal={Advances in large margin classifiers},
  volume={10},
  number={3},
  pages={61--74},
  year={1999},
  publisher={Cambridge, MA}
}

@article{DBLP:journals/pvldb/KangBZ19,
  author       = {Daniel Kang and
                  Peter Bailis and
                  Matei Zaharia},
  title        = {BlazeIt: Optimizing Declarative Aggregation and Limit Queries for
                  Neural Network-Based Video Analytics},
  journal      = {Proc. {VLDB} Endow.},
  volume       = {13},
  number       = {4},
  pages        = {533--546},
  year         = {2019}
}

@article{DBLP:journals/pvldb/LiuGS25,
  author       = {Hanwen Liu and
                  Shashank Giridhara and
                  Ibrahim Sabek},
  title        = {Conformal Prediction for Verifiable Learned Query Optimization},
  journal      = {Proc. {VLDB} Endow.},
  volume       = {18},
  number       = {8},
  pages        = {2653--2666},
  year         = {2025}
}

@article{kang2017noscope,
  title={NoScope: Optimizing Neural Network Queries over Video at Scale},
  author={Kang, Daniel and Emmons, John and Abuzaid, Firas and Bailis, Peter and Zaharia, Matei},
  journal={Proceedings of the VLDB Endowment},
  volume={10},
  number={11},
  year={2017}
}

@article{foygel2021limits,
  title={The limits of distribution-free conditional predictive inference},
  author={Foygel Barber, Rina and Candes, Emmanuel J and Ramdas, Aaditya and Tibshirani, Ryan J},
  journal={Information and Inference: A Journal of the IMA},
  volume={10},
  number={2},
  pages={455--482},
  year={2021},
  publisher={Oxford University Press}
}

@inproceedings{gibbs_adaptive_2021,
	title = {Adaptive {Conformal} {Inference} {Under} {Distribution} {Shift}},
	volume = {34},
	url = {https://proceedings.neurips.cc/paper/2021/hash/0d441de75945e5acbc865406fc9a2559-Abstract.html},
	urldate = {2024-12-09},
	booktitle = {Advances in {Neural} {Information} {Processing} {Systems}},
	publisher = {Curran Associates, Inc.},
	author = {Gibbs, Isaac and Candes, Emmanuel},
	year = {2021},
	pages = {1660--1672},
}

@inproceedings{DBLP:conf/nips/TibshiraniBCR19,
  author       = {Ryan J. Tibshirani and
                  Rina Foygel Barber and
                  Emmanuel J. Cand{\`{e}}s and
                  Aaditya Ramdas},
  title        = {Conformal Prediction Under Covariate Shift},
  booktitle    = {NeurIPS},
  pages        = {2526--2536},
  year         = {2019}
}

@article{barber_conformal_2023,
	title = {Conformal prediction beyond exchangeability},
	volume = {51},
	issn = {0090-5364},
	url = {https://projecteuclid.org/journals/annals-of-statistics/volume-51/issue-2/Conformal-prediction-beyond-exchangeability/10.1214/23-AOS2276.full},
	doi = {10.1214/23-AOS2276},
	number = {2},
	urldate = {2025-06-17},
	journal = {The Annals of Statistics},
	author = {Barber, Rina Foygel and Candès, Emmanuel J. and Ramdas, Aaditya and Tibshirani, Ryan J.},
	month = apr,
	year = {2023},
}

@article{khattab2023dspy,
  title={Dspy: Compiling declarative language model calls into self-improving pipelines},
  author={Khattab, Omar and Singhvi, Arnav and Maheshwari, Paridhi and Zhang, Zhiyuan and Santhanam, Keshav and Vardhamanan, Sri and Haq, Saiful and Sharma, Ashutosh and Joshi, Thomas T and Moazam, Hanna and others},
  journal={arXiv preprint arXiv:2310.03714},
  year={2023}
}

@article{li2023dbet,
  title={dbET: Execution Time Distribution-based Plan Selection},
  author={Li, Yifan and Yu, Xiaohui and Koudas, Nick and Lin, Shu and Sun, Calvin and Chen, Chong},
  journal={Proceedings of the ACM on Management of Data},
  volume={1},
  number={1},
  pages={1--26},
  year={2023},
  publisher={ACM New York, NY, USA}
}

@article{patel2024semantic,
author = {Patel, Liana and Jha, Siddharth and Pan, Melissa and Gupta, Harshit and Asawa, Parth and Guestrin, Carlos and Zaharia, Matei},
title = {Semantic Operators and Their Optimization: Enabling LLM-Based Data Processing with Accuracy Guarantees in LOTUS},
year = {2025},
issue_date = {July 2025},
publisher = {VLDB Endowment},
volume = {18},
number = {11},
issn = {2150-8097},
url = {https://doi.org/10.14778/3749646.3749685},
doi = {10.14778/3749646.3749685},
journal = {Proc. VLDB Endow.},
month = jul,
pages = {4171–4184},
numpages = {14}
}

@article{bordes2013translating,
  title={Translating embeddings for modeling multi-relational data},
  author={Bordes, Antoine and Usunier, Nicolas and Garcia-Duran, Alberto and Weston, Jason and Yakhnenko, Oksana},
  journal={Advances in neural information processing systems},
  volume={26},
  year={2013}
}

@inproceedings{trouillon2016complex,
  title={Complex embeddings for simple link prediction},
  author={Trouillon, Th{\'e}o and Welbl, Johannes and Riedel, Sebastian and Gaussier, {\'E}ric and Bouchard, Guillaume},
  booktitle={International conference on machine learning},
  pages={2071--2080},
  year={2016},
  organization={PMLR}
}

@article{sun2019rotate,
  title={Rotate: Knowledge graph embedding by relational rotation in complex space},
  author={Sun, Zhiqing and Deng, Zhi-Hong and Nie, Jian-Yun and Tang, Jian},
  journal={arXiv preprint arXiv:1902.10197},
  year={2019}
}

\end{document}